\documentclass[aps,amssymb,amsmath,amsfonts,superscriptaddress]{revtex4-2}

\usepackage{graphicx}
\usepackage{bm,bbm}
\usepackage{epstopdf}
\usepackage{amsthm}
\usepackage{amsmath}
\usepackage{color}
\usepackage{hyperref}
\usepackage[T1]{fontenc}
\usepackage[english]{babel}
\usepackage[latin2]{inputenc}
\usepackage{amsopn}
\usepackage{amssymb}
\usepackage{hyperref}
\usepackage{listings}
\usepackage{subfig}
\usepackage{enumerate}
\usepackage{todonotes}
\usepackage{ dsfont }

\usepackage{color}
\usepackage{listings}
\usepackage{enumerate}
\usepackage{todonotes}
\usepackage{multirow}
\usepackage{adjustbox}
\usepackage{mathtools}
\usepackage[ruled, vlined]{algorithm2e}

\newcommand{\eg}{{\emph{e.g. \/}}}
\newcommand{\ie}{{\emph{i.e. \/}}}

\DeclareMathOperator{\tr}{tr}

\newcommand{\N}{\ensuremath{\mathbb{N}}}

\newcommand{\C}{\ensuremath{\mathbb{C}}}

\newcommand{\kraus}[1]{\ensuremath{\mathcal{K}\left( #1 \right)}}

\newcommand{\ket}[1]{\ensuremath{|#1\rangle}}
\newcommand{\bra}[1]{\ensuremath{\langle#1|}}
\newcommand{\ketbra}[2]{\ensuremath{\ket{#1} \! \bra{#2}}}
\newcommand{\proj}[1]{\ensuremath{\ketbra{#1}{#1}}}
\newcommand{\braket}[2]{\ensuremath{\langle{#1}|{#2}\rangle}}
\newcommand{\floor}[1]{\ensuremath{\left\lfloor #1 \right\rfloor}}

\newcommand{\1}{{\rm 1\hspace{-0.9mm}l}}

\newcommand{\EE}{\mathcal{E}}
\newcommand{\XX}{\mathcal{X}}
\newcommand{\MM}{\mathcal{M}}

\newcommand{\YY}{\mathcal{Y}}
\newcommand{\II}{\mathcal{I}}

\newcommand{\DD}{\mathcal{D}}
\newcommand{\FF}{\mathcal{F}}

\newcommand{\CC}{\mathcal{C}}
\newcommand{\PP}{\mathcal{P}}
\newcommand{\RR}{\mathcal{R}}

\renewcommand{\SS}{\mathcal{S}}
\newcommand{\UU}{\mathcal{U}}
\newcommand{\LL}{\mathcal{L}}

\newcommand{\HH}{\mathcal{H}}

\newenvironment{non}[1]
{\noindent\textbf{#1}.\itshape}

\newtheorem{lemma}{Lemma}
\newtheorem{theorem}[lemma]{Theorem}
\newtheorem{corollary}[lemma]{Corollary}
\newtheorem{proposition}[lemma]{Proposition}

\makeatletter
\let\c@algocf\c@lemma
\makeatother

\usepackage{caption}
\def\>{\rangle}
\def\<{\langle}

\begin{document}

\title{On the probabilistic quantum error correction}

\author{Ryszard Kukulski}
\email{rkukulski@iitis.pl}
\affiliation{Institute of Theoretical and Applied Informatics, Polish Academy
	of Sciences, Ba{\l}tycka 5, 44-100 Gliwice, Poland}
\author{{\L}ukasz Pawela}
\affiliation{Institute of Theoretical and Applied Informatics, Polish Academy
	of Sciences, Ba{\l}tycka 5, 44-100 Gliwice, Poland}
\author{Zbigniew Pucha{\l}a}
\affiliation{Institute of Theoretical and Applied Informatics, Polish Academy
	of Sciences, Ba{\l}tycka 5, 44-100 Gliwice, Poland}
\affiliation{Faculty of Physics, Astronomy and Applied Computer Science,
	Jagiellonian University, ul. {\L}ojasiewicza 11,  30-348 Krak{\'o}w, Poland}
\date{ \today}

\begin{abstract}
Probabilistic quantum error correction is an error-correcting 
procedure which uses postselection to determine if the encoded information was 
successfully restored. In this work, we deeply analyze probabilistic version of 
the error-correcting procedure for general noise. We generalized the 
Knill-Laflamme 
conditions for probabilistically correctable errors. We show that 
for some noise channels, we should encode the information into a mixed 
state to maximize the probability of successful error correction. Finally, we 
investigate an advantage of the probabilistic error-correcting procedure over 
the deterministic one. Reducing the probability of successful error correction 
allows for correcting errors generated by a broader class of noise channels. 
Significantly, if the errors are caused by a unitary interaction with an 
auxiliary qubit system, we can probabilistically restore a 
qubit state by using only one additional physical qubit. 
\end{abstract}

\maketitle

\section{Introduction}
Quantum error correction (QEC) is an encoding-decoding procedure that protects 
quantum information from errors arising due to quantum noise. Similarly, as in 
classical computations, this procedure is essential to develop fully 
operational quantum computers \cite{preskill2018quantum}. The theory of QEC, 
initialized by the work of Shor~\cite{shor1995scheme}, covers a wide range of 
coding techniques: Calderbank-Shor-Steane codes~\cite{calderbank1996good, 
steane1996error, steane1996multiple}, stabilizer 
codes~\cite{gottesman1997stabilizer}, topological 
codes~\cite{bombin2006topological}, subsystem codes~\cite{kribs2005unified}, 
entanglement-assisted quantum error-correcting codes~\cite{brun2006correcting, 
brun2014catalytic}, quantum low-density parity-check 
(LDPC) codes~\cite{mackay2004sparse}, quantum maximum distance separable codes 
\cite{huber2020quantum} and many more (for a review see 
\cite{lidar2013quantum}).

In this work, we study a particular QEC procedure called \textit{probabilistic 
quantum error correction} (pQEC) \cite{koashi1999reversing, 
fern2002probabilistic, barberis2010quantum}. To outline how pQEC 
procedure works, let us present an example of classical probabilistic 
error correction. Consider the scenario, when the encoded 
data is harmed by a single bit error, that is with the probability $p \in 
[0,1]$ an arbitrary bit will be flipped. To 
secure 
a one bit of information, we use two physical bits. If we expect that $p\le 
\frac23$, then we can encode $0 \to 00$ and $1 \to 11$. If we receive 
information $00$ at the decoding stage, we are certain the 
encoded message was $0$ (and $1$ for $11$). 
We dismiss the cases $01$ and $10$ as they do not give conclusive answers. 
Otherwise, if $p > \frac23$ it would be beneficial to use encoding $0 \to 00$ 
and 
$1 \to 01$ with the accepting states $10$ and $11$. It is worth mentioning, 
that to secure a one bit of information perfectly, it is necessary to use three 
physical bits, for example $0 \to 000, 1 \to 111$.

Let us return to the quantum case. The heart of pQEC procedure is the 
probabilistic decoding operation~\cite{xiao2013protecting, wang2014protecting}. 
This operation uses a classical postselection to determine if the encoded 
information was successfully restored. The clear drawback is that the 
procedure may fail with some probability. In such case, we should reject the 
output state and ask for a retransmission~\cite{ashikhmin2006fidelity}. 
In the context of QEC, probabilistic 
decoding operations have found application in stabilizer 
codes~\cite{scott2005probabilities} 
especially for iterative probabilistic decoding in LDPC 
codes~\cite{mackay2004sparse, camara2007class, 
	kasai2011quantum}, error decoding~\cite{ashikhmin2000quantum1, 
	ashikhmin2000quantum2} or environment-assisted error 
correction~\cite{wang2014environment}. Moreover, it was noted that they have a 
potential to increase the spectrum of correctable 
errors~\cite{fern2002probabilistic} and 
are useful when the number of qubits is limited~\cite{koashi1999reversing}. 
It is also worth mentioning, they were used with success in other fields of 
quantum information theory, \eg probabilistic 
cloning~\cite{duan1998probabilistic}, 
learning unknown quantum operations~\cite{sedlak2019optimal} or measurement 
discrimination~\cite{puchala2021multiple}. 

Despite the fact that pQEC procedure has been studied in the literature for a 
while, there 
is lack of a formal description of its application for a general noise model. 
In this work, we fill this gap. Inspired by celebrated Knill-Laflamme 
conditions \cite{knill1997theory}, we provide conditions 
(Theorem~\ref{thm-uqec-general}) to check, when probabilistic error 
correction is possible. We discover that optimal error-correcting codes are not 
always generated with the usage of isometric encoding operations. We give an 
explicit example of noise channels family (Section~\ref{sec-example}), such 
that to maximize the probability of successful error correction we need to 
encode the quantum information into a mixed state. Moreover, we discuss the 
advantage of pQEC procedure over the deterministic one 
with a formal statement in Theorem~\ref{thm-nowhere}. We show in 
Theorem~\ref{thm-best-bounds} how to correct noise channels with bounded Choi 
rank. Also, we observe the advantage of pQEC 
procedure for random noise channels, which is presented in 
Theorem~\ref{thm-random}. Finally,
if the errors are caused by a unitary interaction with an 
auxiliary qubit system, we show that it is possible to restore a qubit logical 
state by using only two physical qubits. We present a procedure how to achieve 
this in Algorithm~\ref{alg}.

The rest of the paper is organized as follows. In Section~\ref{sec-prel} we 
introduce the notation and define pQEC protocol. In 
Section~\ref{sec-conditions} we present equivalent 
conditions for probabilistically correctable noise channels. Then, we 
investigate a realization of pQEC 
procedure in Section~\ref{sec-realization}. In Section~\ref{sec-example} we 
present a family of noise channels for which, it is necessary to use mixed 
state encoding to maximize the probability of successful error correction. 
Then, we study an advantage of pQEC procedure in Section~\ref{sec-advantage} 
and Section~\ref{sec-simple}. In Section~\ref{sec-model} we define a 
generalization of pQEC protocol. Finally, we place all proofs in 
Appendix~\ref{sec-app}.

\section{Preliminaries}\label{sec-prel}
\subsection{Mathematical framework}
In this section, we will introduce the notation and recall necessary basic 
facts of quantum information theory. We will denote complex Euclidean spaces 
by symbols $\XX, \YY, \ldots$. The set of linear operators $M: \XX \mapsto \YY$ 
will be written as $\MM(\XX, \YY)$ and $\MM(\XX) \coloneqq \MM(\XX, 
\XX)$. The identity operators will be denoted by $\1_\XX \in \MM(\XX)$. For any 
operator $M \in \MM(\XX,\YY)$ we will consider its vectorization $\ket{M} \in 
\YY 
\otimes \XX$, which is defined as 
\begin{equation}
	\ket{M} \coloneqq \sum_{i=0}^{\dim(\YY) - 1} 
	\ket{i} \otimes M^\top \ket{i},
\end{equation}
where ${\ket{i}}$ are elements of computational basis.
In the space $\MM(\XX)$, we distinguish the 
set of positive semi-definite operators $\PP(\XX)$, the space
of Hermitian operators $\HH(\XX)$ and the set of unitary 
operators $\UU(\XX)$. We use the convention that for non-invertible operator 
$M$, by $M^{-1}$, we denote its Moore-Penrose 
pseudo-inverse~\cite{watrous2018theory}. We consider the 
set of quantum states $\DD(\XX)$, that is, the set of positive semi-definite 
operators with unit trace. We say that a quantum state $\rho$ 
is a pure state if $\mathrm{rank}(\rho)=1$, otherwise, if 
$\mathrm{rank}(\rho)>1$, we say that $\rho$ is a mixed state. The maximally 
mixed state will be denoted by $\rho_{\XX}^* \coloneqq \frac{1}{\dim(\XX)} 
\1_\XX$. 

We also consider transformations between linear operators. We denote by 
$\II_\XX: \MM(\XX) \mapsto \MM(\XX)$ the identity map. 
Let us define the set of quantum subchannels $s\CC(\XX,\YY)$ 
\cite{hellwig1969pure}. 
A quantum subchannel $\Phi \in 
s\CC(\XX,\YY)$ is a linear map $\Phi: \MM(\XX) \mapsto \MM(\YY)$, which is 
completely positive~\cite[Theorem 2.22]{watrous2018theory}, \ie
\begin{equation}
	(\Phi \otimes \II_\XX)(Q) \in \PP(\YY \otimes \XX) \quad \mbox{for any } Q 
	\in \PP(\XX \otimes \XX)
\end{equation} and trace 
non-increasing 
\begin{equation}
	\tr(\Phi(\rho)) \le 1 \quad \mbox{for any } \rho \in \DD(\XX).
\end{equation}
In particular, the subchannel $\Phi$ which is trace preserving, \ie 
\begin{equation}
\tr(\Phi(\rho)) = 1 \quad \mbox{for any } \rho \in \DD(\XX)	
\end{equation}
will be called a quantum channel. We denote by 
$\CC(\XX, \YY)$ the set of quantum channels $\Phi: \MM(\XX) \mapsto \MM(\YY)$. 
We will also use the following notation, $s\CC(\XX) \coloneqq s\CC(\XX,\XX)$ 
and $\CC(\XX) \coloneqq \CC(\XX,\XX)$.

In this work, we will consider the following representations of subchannels:
\begin{itemize}
	\item Kraus representation: Each subchannel $\Phi \in s\CC(\XX,\YY)$ can be 
	defined by a collection of Kraus operators $(K_i)_{i=1}^r \subset 
	\MM(\XX,\YY)$, 	such that $\Phi(X) = \sum_{i=1}^r K_i X K_i^\dagger$ for $X 
	\in 
	\MM(\XX)$ and $r \in \N$. The operators $K_i$ satisfy the condition 
	$\sum_{i=1}^r K_i^\dagger 
	K_i \le \1_\XX$. We say that the subchannel $\Phi$ is given in a canonical 
	Kraus representation $(K_i)_{i=1}^r$, if it holds that $\tr(K_j^\dagger 
	K_i) \propto  \delta_{ij}$ and $K_i \neq 0$ for each $i \le r$. To 
	represent 
	the subchannel $\Phi$ by its Kraus 
	representation $(K_i)_{i=1}^r$, we introduce the notation $\mathcal{K}: 
	\MM(\XX,\YY)^{\times r} \mapsto s\CC(\XX,\YY) $ given by 
	$\Phi = \kraus{(K_i)_{i=1}^{r}}$.
	\item Choi-Jamio{\l}kowski representation: Each subchannel $\Phi \in 
	s\CC(\XX, 
	\YY)$ can be uniquely described by its Choi-Jamio{\l}kowski operator 
	$J(\Phi) \in \MM(\YY \otimes \XX)$, which is defied as $J(\Phi) \coloneqq 
	(\Phi 
	\otimes \II_\XX)(\proj{\1_\XX})$. The rank of $J(\Phi)$ is called the Choi 
	rank and it determines the minimal number $r$ of Kraus operators $K_i$ 
	needed to 
	describe $\Phi$ in the Kraus form $\Phi = \kraus{(K_i)_{i=1}^r}$. 
	Therefore, if the 
	Kraus representation $(K_i)_{i=1}^r$ is canonical, then $r = 
	\mathrm{rank}(J(\Phi))$. 
	\item Stinespring representation: By the Stinespring Dilatation Theorem any 
	subchannel $\Phi \in s\CC(\XX,\YY)$ 
	can be defined as $\Phi(X) = \tr_2 
	\left(A X A^\dagger 
	\right)$ for $X \in \MM(\XX)$, where $A \in \MM(\XX, \YY \otimes \C^r)$ and 
	$\tr_2$ is the partial trace over the second subsystem $\C^r$. 
	The minimal dimension $r$ of the auxiliary system is equal to the Choi 
	rank. In 
	particular, for $\Phi \in \CC(\XX)$, the Stinespring representation of 
	$\Phi$ can be written in the form $\Phi(X) = \tr_2 
	\left(U (X \otimes \proj{\psi}) U^\dagger\right)$, where $\proj{\psi} \in 
	\DD(\C^r)$ and $U \in \UU(\XX \otimes \C^r)$.
	
\end{itemize}

\subsection{Problem formulation}
In this work, we consider the following procedure of probabilistic quantum 
error correction. We are given a noise channel $\EE \in \CC(\YY)$ and a 
Euclidean 
space $\XX$. The goal of pQEC is to choose an appropriate encoding operation 
$\SS 
\in s\CC(\XX, \YY)$ and decoding operation $\RR \in s\CC(\YY, \XX)$, such that 
for any state $\rho \in \DD(\XX)$ we have $\RR\EE \SS(\rho) \propto \rho.$ 
In this protocol, the pair $(\SS, \RR)$ represents the error-correcting 
scheme and the quantity $\tr \left( \RR \EE \SS (\rho) \right)$ 
represents the probability of successful error correction. This protocol may 
fail with the probability $1-  \tr \left( 
\RR \EE \SS 
(\rho) \right)$. In such a case, the output state is rejected. To exclude a 
trivial, null strategy, we add the constrain that a valid error-correcting 
scheme must satisfy $\tr(\RR \EE \SS(\rho)) > 0$ for any $\rho \in \DD(\XX)$.

In this set-up, the probability of successful error correction does not depend 
on the input state $\rho$ (see Lemma~\ref{lem-constant-p} in 
Appendix~\ref{proof-lem-constant-p}). We use this fact to standardize the 
definition of pQEC. From now, we say that $\EE \in \CC(\YY)$ 
is probabilistically correctable for $\XX$, if there exists an error-correcting 
scheme $(\SS, \RR)$ such that
\begin{equation}
	0 \neq \RR \EE \SS \propto \II_\XX.
\end{equation}
We say that $\EE$ is correctable perfectly if $\RR \EE \SS = \II_\XX$. In this 
work, we will be particularly interested in error-correcting schemes 
$(\SS,\RR)$, which maximize the probability of 
success for given $\EE$ and $\XX$.

\section{Probabilistic quantum error correction}\label{sec-conditions}

To inspect pQEC procedure, first, we should state conditions which 
determine when given noise channel is probabilistically correctable. For 
deterministic QEC, such conditions have been known for a long time and in the 
literature as the Knill-Laflamme conditions~\cite{knill1997theory}. Let $\EE  = 
\kraus{(E_i)_i} \in \CC(\YY)$ be a given noise channel. Then, according to the 
Knill-Laflamme Theorem, $\EE$ is 
perfectly correctable for $\XX$ if and only if
\begin{equation}\label{eq-txt-kl}
	S^\dagger E_j^\dagger E_i S \propto \1_\XX
\end{equation}
for all $i,j$ and some isometry operator $S \in \MM(\XX, \YY)$. In the 
following 
theorem we generalize the above, to cover probabilistically correctable noise 
channels.
\begin{theorem}[Equivalent conditions for pQEC]\label{thm-uqec-general}
	Let $\EE = \kraus{(E_i)_i} \in \CC(\YY)$. The following conditions are
	equivalent:
	\begin{enumerate}[(A)]
		\item There exist error-correcting scheme $(\SS, \RR) 
		\in s\CC(\XX,\YY) \times s\CC(\YY,\XX)$ and $p>0$ such that
		\begin{equation}
			\RR \EE \SS  = p \II_\XX.
		\end{equation}
		\item There exist $S
		= \kraus{(S_k)_k} \in s\CC(\XX,\YY)$ and $R \in \PP(\YY)$, such that $R 
		\leq \1_\YY$, for which it holds
		\begin{equation}
			\kraus{ \left( \sqrt{R} E_i S_k \right)_{i,k} } = \kraus{(A_i)_i}: 
			\quad A_i \not = 0, A_j^\dagger A_i \propto
			\delta_{ij} \1_\XX.
		\end{equation}
		\item There exist $\SS =
		\kraus{(S_k)_k} \in s\CC(\XX,\YY)$, $R \in \PP(\YY)$, such that $R \leq 
		\1_\YY$ and a matrix $M= [M_{jl,ik}]_{jl,ik} 
		\not=0$, for which it
		holds
		\begin{equation}
			\forall_{i,j,k,l} \quad S_l^\dagger E_j^\dagger R E_i S_k = M_{jl, 
			ik} \1_\XX.			
		\end{equation}
		\item There exist $S_* \in \MM(\XX,\YY)$ and $R_* \in \MM(\YY, 
		\XX) $ such that
		\begin{equation}\label{eq-thm-1d}
			\forall_i \quad R_* E_i S_* \propto \1_\XX
	\end{equation}
	and there exists $i_0$, for which it holds $R_* E_{i_0} S_* \neq 0$.
	\end{enumerate}

Moreover, if point $(A)$ holds for $\SS = \kraus{(S_k)_k}$ and $\RR = 
\kraus{(R_l)_l}$, then $R \in \PP(\YY)$ from points $(B)$ and $(C)$ 
can be chosen to satisfy $	R = \sum_l R_l^\dagger R_l.$ It also holds that
$R_l E_i S_k \propto \1_\XX$  for any $i,k,l$.
\end{theorem}
The proof of Theorem~\ref{thm-uqec-general} is presented in 
Appendix~\ref{proof-thm-uqec-general}. Let us discuss the meaning of the 
conditions stated in Theorem~\ref{thm-uqec-general}. The condition $(B)$ 
presents a general form of probabilistically correctable 
noise channels $\EE$. Such channels, after applying 
post-processing $\sqrt{R}$ behave as mixed isometry operations. They 
hide parts of an initial quantum information on orthogonal subspaces. 
The condition $(C)$ may be used to calculate the maximum value of the 
probability $p$ of successful error correction. For $r = 
\mathrm{rank}(J(\EE))$ and $d = \dim(\XX)$, $s =\dim(\YY)$ we can introduce the 
optimization procedure:
\begin{equation*}
	\begin{split}
		\text{maximize:}\quad &
		\tr(M)
		\\[2mm]
		\text{subject to:}\quad & 
		S_l^\dagger E_j^\dagger R E_i S_k = M_{jl, ik} \1_\XX, \quad 
		\forall_{i,j,k,l}\\
		& 0 \le R \leq \1_\YY,\\
		& \sum_k S_k^\dagger S_k \le 
		\1_\XX\\
		&R \in \MM(\YY),(S_k)_{k = 1}^{ds}\subset \MM(\XX, \YY), M \in \MM(\C^r 
		\otimes \C^{ds})\\ 
	\end{split}
\end{equation*}
Moreover, one may get the form of a recovery subchannel $\RR$ based on $R, 
\SS = \kraus{(S_k)_k}$ and $M$ obtained from this optimization in the following 
way (see Appendix~\ref{proof-thm-uqec-general}):
\begin{enumerate}
	\item Let $M = U^\dagger D U$ be the spectral decomposition of $M$.
	\item Define $A_{ii'} = \sum_{a,b} \overline{ U_{ii', ab}} \sqrt{R} E_{a} 
	S_b $.
	\item For each $A_{ii'} \neq 0$ define $\alpha_{ii'}: \, \, A_{ii'}^\dagger 
	A_{ii'} = \alpha_{ii'} \1_{\XX}$.
	\item The recovery subchannel is given as $\RR= \kraus{ \left(
	\alpha_{ii'}^{-1/2} A_{ii'}^\dagger \sqrt{R} \right)_{i,i'} }$.
\end{enumerate}
Finally, the condition $(D)$ gives us a simple method to check if $\EE= 
\kraus{(E_i)_{i=1}^r}$ is probabilistically correctable for $\XX$. Let us 
compare the point $(D)$ with Knill-Laflamme conditions. The latter, is a 
constraint satisfaction problem with $r^2$ quadratic constrains $S^\dagger 
E_j^\dagger E_i S \propto \1_\XX$ for the variable $S \in \MM(\XX, \YY)$, which 
satisfies $S \neq 0$. The parameters $E_j^\dagger E_i$ constitute 
$^\dagger-$closed algebra $\mathcal{A}$, such that $\1_\YY \in \mathcal{A}$. In 
comparison, the conditions in the point $(D)$ represent a constraint 
satisfaction problem with $r$ bilinear constrains $RE_iS \propto \1_\XX$ for 
the 
variables $S \in \MM(\XX,\YY)$ and $R \in \MM(\YY, \XX)$. Additionally, it must 
hold $RE_{i_0}S \neq 0$ for some $i_0 \in \{1,\ldots,r\}$. In this problem, the 
parameters $E_i$ are arbitrary operators from $\MM(\YY)$, which satisfy 
$\mathrm{span}\left(\mathrm{im}(E_i^\dagger): i=1,\ldots,r\right) = \YY$ 
(although a stronger condition holds $\sum_i E_i^\dagger E_i = \1_\YY$, we will 
see in Section~\ref{sec-advantage}, it is more convenient to use the weaker 
version).

\section{Realization of pQEC procedure}\label{sec-realization}

In this section, we will investigate the form of error-correcting 
scheme $(\SS, \RR)$ which provides the maximal probability of successful error 
correction. For perfectly correctable noise channels, the encoding 
$\SS$ can be realized by the isometry channel. This observation meaningfully 
reduces the complexity of finding error-correcting schemes -- it is enough to 
consider a vector representation of pure states. Inspired by that, we ask if 
a similar behavior occurs in the probabilistic quantum error correction. The 
following proposition gives us some insight in the form of encoding and 
decoding.

\begin{proposition}\label{prop-realization}
For a given channel $\EE \in \CC(\YY)$, let us fix an error-correcting 
scheme $(\SS, \RR) \in s\CC(\XX,\YY) \times s\CC(\YY,\XX)$ such that $
\RR \EE \SS = p \II_\XX$, for some $p > 0.$ Then, the following holds:
\begin{enumerate}[(A)]
\item There exist $\widetilde \SS \in \CC(\XX,\YY)$ and $\widetilde \RR \in 
s\CC(\YY,\XX)$ such that $\widetilde \RR \EE \widetilde \SS = p \II_\XX.$
\item If $\RR \in \CC(\YY,\XX)$, then there exists $\widetilde \SS = 
\kraus{(\widetilde S)} 
\in \CC(\XX,\YY)$ such that $\RR \EE \widetilde \SS = \II_\XX.$
\item If $p=1$, then there exist $\widetilde \SS = \kraus{(\widetilde S)} 
\in \CC(\XX,\YY)$ and $\widetilde \RR \in \CC(\YY,\XX)$ such that $
\widetilde \RR \EE \widetilde \SS = \II_\XX.$
\end{enumerate}
\end{proposition}
The proof of Proposition~\ref{prop-realization} is presented in 
Appendix~\ref{proof-prop-realization}. 
\begin{figure}[h!]
	\includegraphics{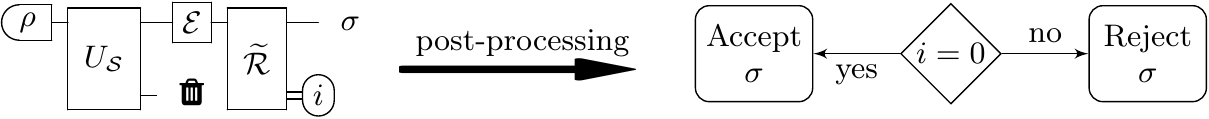}
	\caption{ Schematic realization of pQEC procedure for the 
	noise channel $\EE$. 
	}\label{fig-realization}
\end{figure}
We may use Proposition~\ref{prop-realization} $(A)$ to state a realization of 
pQEC procedure (see Figure~\ref{fig-realization}). For a given noise 
channel $\EE \in \CC(\YY)$ let $(\SS, \RR) \in \CC(\XX,\YY) \times 
s\CC(\YY,\XX)$ be an error-correcting scheme for which $\RR\EE\SS = p \II_\XX$, 
where $p > 0$. The encoding channel $\SS$ can be realized using the 
Stinespring representation given in the form $\SS(X) = \tr_2\left(U_\SS X 
U_\SS^\dagger\right)$. The state is then sent through 
$\EE$. The decoding subchannel 
$\RR \in s\CC(\YY, \XX)$ can be realized by implementing the channel 
$\widetilde \RR \in \CC(\YY, \XX \otimes \C^2)$ given in the form $\widetilde 
\RR (Y) = \RR(Y) \otimes \proj{0} + \Psi(Y) \otimes \proj{1}$, where $\Psi \in 
s\CC(\YY,\XX)$ such that $(\RR + \Psi) \in \CC(\YY, 
\XX)$. In summary, the output of the whole procedure consists of a 
quantum state $\sigma \in \DD(\XX)$ and a classical label $i \in \{0,1\}$. If 
the label $i = 0$ is obtained, we know that $\sigma \propto \RR\EE\SS(\rho) 
= p \rho$, and hence, the output state can be 
accepted. Otherwise, if $i=1$, the output state $\sigma \propto \Psi \EE \SS 
(\rho)$ should be rejected, as in general it may differ from $\rho$.

In 
Proposition~\ref{prop-realization} $(C)$, we observed that using non-isometric 
channels $\SS$ or formal subchannels $\RR$ for perfectly correctable noise 
channels provides no advantage. Moreover, according to 
Theorem~\ref{thm-uqec-general} $(D)$, to 
predict if a noise channel is probabilistically correctable, we may consider 
only 
single Kraus encoding operations. However, among all conditions presented in 
Proposition~\ref{prop-realization} there is 
no condition, which in general allows us to restrict our attention to an 
isometry channel realization of $\SS$. Indeed, there is a class of 
noise channels $\EE$ for which, in order to maximize the probability $p$ of 
successful 
error correction, we need to consider a general channel realization of 
$\SS$. Paraphrasing, to obtain the best performance, we have to encode the 
initial 
state $\proj{\psi} \in \DD(\XX)$ into the mixed state $\SS(\proj{\psi})$.
In Section~\ref{sec-example} we will present a family of noise channels for 
which it is necessary to use mixed state encoding.

\section{Need for mixed state encoding}\label{sec-example}

In this section, we provide an example of a parametrized family of noise 
channels $ \{\EE_R\}_R$ for which the mixed state 
encoding 
improves the probability of successful error correction. In our example we 
assume that $\XX = \C^2$ and $\YY = \C^4$. For each 
$R \in \PP(\C^4)$ satisfying $R \le \1_{\C^4}$ let us define a noise channel 
$\EE_R 
\in \CC(\C^4)$ given by the equation
\begin{equation}\label{eq-param-noise}
	\begin{split}
		\EE_R(Y) = \proj{0} \otimes \tr_1\left( \sqrt{R} Y \sqrt{R}\right) + 
		\proj{1} \otimes \tr\left([\1_{\C^4} - R] Y\right) \rho_2^*.
	\end{split}
\end{equation}
We define the optimal probability $p_0$ of successful error correction as
\begin{equation}
p_0(R) \coloneqq \max\left\{p: \,\, \RR \EE_R \SS = p \II_{\C^2}, \,\, (\SS, 
\RR) 
\in s\CC(\C^2, \C^4) \times s\CC(\C^4,\C^2)\right\}.
\end{equation}
We also define the optimal probability $p_1$ of 
successful error correction restricted to the pure state encoding:
\begin{equation}
p_1(R) \coloneqq \max\left\{p:\,\,  \RR \EE_R \SS = p \II_{\C^2}, \,\, \SS = 
\kraus{(S)}, \,\, (\SS, \RR) 
\in s\CC(\C^2, \C^4) \times s\CC(\C^4,\C^2)\right\}.
\end{equation}
Our claim, which we will present later, is that there exists a family of 
operators $R$ for which 
$p_0(R) > p_1(R)$. 

We start with the following 
lemma, where we show the optimal error-correcting 
scheme $(\SS, \RR)$ and a simplified version of the maximization 
problem $p_0(R)$. 

\begin{lemma}\label{lem-example-gen}
Let $R \in \PP(\C^4)$ and $R \le \1_{\C^4}$. Define $\Pi_R$ as a projector on 
the 
support of $R$. For $\EE_R$ defined in 
Eq.~\eqref{eq-param-noise} we have the following simplified form of the 
maximization problem $p_0(R)$:
\begin{equation}\label{eq-param-noise-lem}
	p_0(R) = \max\left\{ \tr(P): \,\, P \in \PP(\C^2), \tr_1\left(R^{-1} (P 
	\otimes \1_{\C^2})\right) \le \1_{\C^2}, 
	\,\, 
	\forall_{X \in \MM(\C^2)} \, \,\Pi_R (P \otimes X) \Pi_R = P \otimes X 
	\right\}.
\end{equation}
An optimal scheme $(\SS, \RR)$ which achieves the probability $p_0(R)$, that is 
$\RR \EE_R \SS = p_0(R) \II_{\C^2}$, can be taken as
\begin{equation}
\begin{split}
\SS(X) &= \sqrt{R}^{-1} (P \otimes X) \sqrt{R}^{-1},\\
\RR(Y) &= \tr_1\left(Y\left(\proj{0} \otimes \1_{\C^2}\right) \right),
\end{split}
\end{equation}	
where $P$ is an argument maximizing $p_0(R)$ in Eq.~\eqref{eq-param-noise-lem}. 
Moreover, if there exists another optimal scheme $(\widetilde\SS, 
\widetilde\RR)$, that is $\widetilde\RR \EE_R \widetilde \SS = p_0(R) 
\II_{\C^2}$, 
then $\mathrm{rank}(J(\SS)) \le \mathrm{rank}(J(\widetilde{\SS}))$.
\end{lemma}
The proof of Lemma~\ref{lem-example-gen} is presented in 
Appendix~\ref{proof-lem-example-gen}. Let us separately consider two cases: 
$\mathrm{rank}(R) < 4$ and $\mathrm{rank}(R) = 4$. The first one will 
be discussed briefly as it will not support our claim. 
\begin{corollary}\label{cor-example-rank-123}
Let us take $R \in \PP(\C^4)$ such that $R \le \1_{\C^4}$ and $\mathrm{rank}(R) 
< 
4$. Define $\Pi_R$ as a projector on 
the 
support of $R$. For the noise channel defined in Eq.~\eqref{eq-param-noise} we 
have $p_0(R) 
= 
p_1(R)$. Moreover, it holds
\begin{equation}
p_0(R) = 
\begin{cases}
0, &\mathrm{rank}(R) \le 1,\\
0, &\mathrm{rank}(R) = 2, \Pi_R \neq 
\proj{\psi} \otimes \1_{\C^2}, \ket{\psi} \in \C^2,\\
\|\tr_1\left(R^{-1} 
(\proj{\psi} 
\otimes \1_{\C^2})\right)\|_\infty^{-1}, &\mathrm{rank}(R) = 2, 
\Pi_R = 
\proj{\psi} \otimes \1_{\C^2}, \ket{\psi} \in \C^2,\\
0, &\mathrm{rank}(R) = 3, \Pi_R = 
\1_{\C^4} - \proj{\alpha}, \C^4 \ni \ket{\alpha} \mbox{ is 
entangled},\\
\|\tr_1\left(R^{-1} 
(\proj{\psi} 
\otimes \1_{\C^2})\right)\|_\infty^{-1}, &\mathrm{rank}(R) = 3, \Pi_R = 
\1_{\C^4} - 
\proj{\psi^\perp} \otimes 
\proj{\phi}, \ket{\psi^\perp},\ket{\phi} \in \C^2, \proj{\psi} \in \DD(\C^2),
\end{cases}
\end{equation}
where $R^{-1}$ denotes Moore-Penrose 
pseudo-inverse.
\end{corollary}
The proof of Corollary~\ref{cor-example-rank-123} is presented in 
Appendix~\ref{proof-cor-example-rank-123}. 

In the case when the operator $R$ is invertible, the situation is more 
interesting. Let us focus on $p_0(R)$ obtained in 
Eq.~\eqref{eq-param-noise-lem}. As $\Pi_R = \1_{\C^4}$, the 
equation $\Pi_R (P \otimes X) \Pi_R = P \otimes X$ is always satisfied. We can 
take $P = \tr(P) \rho$, for $\rho \in \DD(\C^2)$. The inequality  $\tr(P) 
\tr_1\left(R^{-1} (\rho \otimes \1_{\C^2})\right) \le \1_{\C^2}$ is equivalent 
to $\tr(P) \le 
\|\tr_1\left(R^{-1} 
	(\rho\otimes \1_{\C^2})\right)\|_\infty^{-1}$. Hence, we 
get
\begin{equation}\label{eq-example-opt-prob}
p_0(R) = \max \left\{\|\tr_1\left(R^{-1} 
	(\rho
	\otimes \1_{\C^2})\right)\|_\infty^{-1}: \,\, \rho \in \DD(\C^2) \right\}.
\end{equation}

To calculate $p_1(R)$ it will be sufficient to add the constraint $\SS = 
\kraus{(S)}$. According to Lemma~\ref{lem-example-gen} the optimal $\SS$ is of 
the form  $\SS(X) = \sqrt{R}^{-1} (P \otimes X) \sqrt{R}^{-1}$. As $R$ is 
invertible, $\SS = 
\kraus{(S)}$ if and only if $P = \proj{\psi}$ for some $\ket{\psi} \in \C^2$.
Then, we have
\begin{equation}
p_1(R) = \max \left\{\|\tr_1\left(R^{-1} 
	(\proj{\psi}
	\otimes \1_{\C^2})\right)\|_\infty^{-1}: \,\, \proj{\psi} \in \DD(\C^2) 
	\right\}.
\end{equation}
 \begin{proposition}\label{prop-mixed-encoding}
Let us define an unitary matrix $U \in \UU(\C^4)$ which columns form the magic 
basis 
\cite{hill1997entanglement}
\begin{equation}
	U = \frac{1}{\sqrt{2}} \left[\begin{array}{cccc}
		1 & 0 & 0 & i\\
		0 & i & 1 & 0\\
		0 & i & -1& 0\\
		1 & 0 & 0 &-i
	\end{array}\right].
\end{equation}
Let us also define a diagonal operator $D(\lambda) \coloneqq 
\mathrm{diag}^\dagger 
\left(\lambda\right)$, which is parameterized by 
a $4-$dimensional real vector $\lambda = (\lambda_1, \lambda_2, \lambda_3, 
\lambda_4)$, for which it holds $0 < \lambda_i \le 1$. For $R = U D(\lambda) 
U^\dagger$ and the noise channel $\EE_R$ defined in 
Eq.~\eqref{eq-param-noise} we have
\begin{equation}
\begin{split}
p_0(R) &=  \frac{4}{\tr(R^{-1})},\\
p_1(R) &= \frac{4}{\tr(R^{-1}) + \min\left\{ \left|\frac{1}{\lambda_1} - 
	\frac{1}{\lambda_2} 
	- \frac{1}{\lambda_3} + \frac{1}{\lambda_4}\right|, 
	\left|\left|\frac{1}{\lambda_1}-\frac{1}
	{\lambda_4}\right|
	-\left|\frac{1}{\lambda_2}-
	\frac{1}{\lambda_3}\right|\right| \right\}}.
\end{split}
\end{equation}
\end{proposition}
The proof of Proposition~\ref{prop-mixed-encoding} is presented in 
Appendix~\ref{proof-prop-mixed-encoding}. We can clearly see that in the case 
$\mathrm{rank}(R) = 4$, there are operators $R$, for which the mixed state 
encoding improves the probability of successful error correction over the pure 
state encoding, $p_0(R) > p_1(R)$. In general, the maximization problem 
in Eq.~\eqref{eq-example-opt-prob} intuitively supports the inequality $p_0(R) 
> 
p_1(R)$. The function $\rho \mapsto \|\tr_1\left(R^{-1} (\rho \otimes 
\1_{\C^2})\right)\|_\infty$ is convex, so it is possible, that the minimal 
value of it will be achieved for some mixed state $\rho$. 
We observed such behavior in Proposition~\ref{prop-mixed-encoding} for $R$ 
given in the spectral decomposition $R = U D(\lambda) U^\dagger$. The 
introduced family of noise channels is parameterized by a $4-$dimensional 
vector 
$\lambda 
= (\lambda_1, \ldots, 
\lambda_4)$, such that $\lambda_i 
\in (0,1]$. For almost all such $\lambda$ we have $p_0(R) > p_1(R)$. The only 
exception is 
the $3-$dimensional subset defined by the relation 
\begin{equation}
\frac{1}{\lambda_1} + 
\frac{1}{\lambda_4}= \frac{1}{\lambda_2} + 
\frac{1}{\lambda_3} \vee \left|\frac{1}{\lambda_1}-\frac{1}
{\lambda_4}\right| = \left|\frac{1}{\lambda_2}-
\frac{1}{\lambda_3}\right|,
\end{equation}
which describes the situation, when the pure state encoding match the mixed 
state encoding, $p_0(R) = p_1(R)$. In an extremal case, \eg for $\lambda = 
(\frac{1}{2N}, \frac{1}{2}, 
\frac{1}{2}, 
\frac{1}{2})$, $N \in \N$, we get $p_1(R) = \frac{1}{N + 1}$ and $p_0(R) = 
\frac{2}{N + 3}$. Especially, when $N \to \infty$ the mixed state encoding 
provides the advantage, $p_0(R) / p_1(R) \to 2$.

The family of parameters $R$ introduced in 
Proposition~\ref{prop-mixed-encoding} is not the 
only one for which the minimum value of $ \|\tr_1\left(R^{-1} 
(\rho\otimes \1_{\C^2})\right)\|_\infty$ is achieved for mixed state $\rho$. 
Let 
$R^{-1} \propto (\II_{\C^2} \otimes \Phi)(\proj{\1_{\C^2}})$ for some $\Phi \in 
\CC(\C^2)$. 
Then, $\|\tr_1\left(R^{-1} (\rho \otimes \1_{\C^2})\right)\|_\infty \propto \| 
\Phi(\rho^\top)\|_\infty$. Therefore, the value of $p_0(R)$ is one-to-one 
related with the 
maximum value of the output 
min-entropy of the channel $\Phi$ (see for instance \cite{muller2013quantum}). 
Especially, we can see, if the image of the 
Bloch 
ball under $\Phi$ is a three dimensional ellipsis and contains the maximally 
mixed state $\rho_2^*$ in its interior, 
then the mixed state encoding provides benefits.

Finally, the noise channel $\EE_R$ defined for $R$ from 
Proposition~\ref{prop-mixed-encoding} is perfectly 
correctable for $\XX = \C^2$ if and only if $R = \1_{\C^4}$. Interestingly, 
this 
suggests that 
perfectly 
correctable noise channels may constitute only a small subset of 
probabilistically 
correctable noise channels. This behavior will be the object of our 
investigation in 
the next section.

\section{Advantage of pQEC procedure}\label{sec-advantage}

The goal of this section is to show that pQEC procedure corrects a 
wider class of noise channels than the QEC procedure based on Knill-Laflamme 
conditions Eq.~\eqref{eq-txt-kl}. For any Euclidean spaces $\XX,\YY$ let 
us define two families of noise channels; these which are probabilistically 
correctable 
for $\XX$ as $\xi(\XX, \YY)$, and these which are correctable perfectly for 
$\XX$ as $\xi_1(\XX,\YY)$:
\begin{equation}
	\begin{split}
	\xi(\XX,\YY) &\coloneqq \{ \EE \in \CC(\YY): \, \, \exists_{(\SS, \RR) \in 
			s\CC(\XX,\YY) \times s\CC(\YY,\XX)} \,\,  0 \neq \RR \EE \SS 
			\propto 
			\II_\XX\},\\
	\xi_1(\XX,\YY) &\coloneqq \{ \EE \in \CC(\YY): \, \, \exists_{(\SS, \RR) 
	\in s\CC(\XX,\YY) \times s\CC(\YY,\XX)} \,\,  \RR \EE \SS = 
	\II_\XX \}.		
	\end{split}
\end{equation}
We begin our analysis with some observations.
\begin{proposition}\label{prop-prop}
For any $\XX$, $\YY$ we have the following 
properties:
\begin{enumerate}[(A)]
\item $\xi_1(\XX,\YY) \subset \xi(\XX,\YY),$
\item If $\dim(\XX) > \dim(\YY)$, then $\xi(\XX,\YY) = \emptyset,$
\item If $\dim(\XX) \le \dim(\YY)$, then $\xi_1(\XX,\YY) \neq \emptyset,$ 
\item If $ \dim(\XX) = \dim(\YY)$, then $\xi_1(\XX,\YY) = \xi(\XX,\YY).$
\end{enumerate}
\end{proposition} 
The proof of Proposition~\ref{prop-prop} is presented in 
Appendix~\ref{proof-prop-prop}. We see that if $\dim(\XX) = \dim(\YY)$, then 
there is no need to consider pQEC procedure. The situation changes if we 
encode the initial information into a larger space, $\dim(\YY) > \dim(\XX)$. In 
the following theorem, we will show that $\xi_1(\XX,\YY) \subsetneqq 
\xi(\XX,\YY)$ for $\dim(\YY) > \dim(\XX)$.  

\begin{theorem}\label{thm-nowhere}
Let $\XX$ and $\YY$ be Euclidean spaces for which $\dim(\XX) < \dim(\YY)$. 
Then, 
the set $\xi_1(\XX,\YY)$ is a nowhere dense subset of $\xi(\XX,\YY)$. 
\end{theorem}
The proof of Theorem~\ref{thm-nowhere} is presented in 
Appendix~\ref{proof-thm-nowhere}.

\subsection{Choi rank of correctable noise channels}

Intensity of a noise channel $\EE$ can be connected with its Choi rank $r = 
\mathrm{rank}(J(\EE))$. Given $\EE$ in the Stinespring form, the Choi rank 
describes the dimension of an environment system which unitarily interacts 
with the encoded information. If the interaction is the weakest ($r=1$) we 
deal with unitary noise channels, which are always perfectly correctable. The 
strongest 
interaction ($r = \dim(\YY)^2$) is a property of hardly correctable noise 
channels. 
For example, the maximally depolarizing channel $\EE(Y) = \tr(Y) \rho_{\YY}^*$, 
which can not be corrected, has the maximal Choi rank. In 
the following theorem, we investigate the maximum Choi rank of 
probabilistically correctable noise channels $\xi(\XX,\YY)$ and compare it with 
the 
maximum Choi rank for $\xi_1(\XX,\YY)$.
\begin{theorem}\label{thm-rank}
Let $\XX$ and $\YY$ be some Euclidean spaces such that $\dim(\YY) \ge 
\dim(\XX)$. 
The following relations hold:
\begin{equation}
\begin{array}{llll}
(A) && \max \left\{ \mathrm{rank}(J(\EE)): \EE \in \xi_1(\XX, \YY) \right\} &= 
\dim(\YY)^2 
	- \dim(\YY) \dim(\XX) + \floor{\frac{\dim(\YY)}{\dim(\XX)}},\\
(B) && \max \left\{ \mathrm{rank}(J(\EE)): \EE \in \xi(\XX, \YY) \right\} &= 
\dim(\YY)^2 
	- \dim(\XX)^2 + 1.
\end{array}
\end{equation}
\end{theorem}
The proof of Theorem~\ref{thm-rank} is presented in 
Appendix~\ref{proof-thm-rank}. In Proposition~\ref{prop-prop} we showed 
that if $\dim(\XX) = \dim(\YY)$, then the pQEC procedure gives us no 
advantage. Indeed, the only reversible noise channels, 
in this case, are unitary noise channels. In the language of Choi 
rank, that means, if the Choi rank of a noise channel is equal to one, 
then it can be corrected. We can ask, what is the maximum value of $r \in \N$, 
such that all noise channels which Choi rank is less or equal $r$, can be 
corrected 
perfectly or probabilistically, respectively. Formally speaking, for any $\XX$ 
and $\YY$ we define the following quantities:
\begin{equation}\label{eq-def-r}
	\begin{array}{lrll}
		&r_1(\XX,\YY) \coloneqq & \max \left\{r \in \N: \quad \forall_{\EE \in 
		\CC(\YY)} 
		\,\,  
		\mathrm{rank}(J(\EE)) \le r \implies \EE \in \xi_1(\XX, \YY) \right\} 
		,\\
		&r(\XX,\YY) \coloneqq & \max \left\{r \in \N : \quad \forall_{\EE \in 
		\CC(\YY)} \,\,  
		\mathrm{rank}(J(\EE)) \le r \implies \EE \in \xi(\XX, \YY) \right\}.
	\end{array}
\end{equation}
The quantity $r_1(\XX, \YY)$ for a general noise model was studied in 
\cite{knill2000theory, chiribella2011quantum}. The authors of 
\cite{knill2000theory} calculated a lower bound for $r_1(\XX, \YY)$ by using a 
technique of noise 
diagonalization along with Tverberg's theorem. They obtained the following 
result
	\begin{equation}
		\max\left\{r \in \N: \dim(\XX) \le \frac{\left\lceil 
		\frac{\dim(\YY)}{r^2} 
		\right\rceil+r^2}{r^2+1}\right\} \le r_1(\XX,\YY).
	\end{equation}
It implies that $\left\lfloor \sqrt[4]{\frac{\dim(\YY)}{\dim(\XX)}} 
\right\rfloor \le 
r_1(\XX,\YY).$ On the other 
hand, by using the Quantum packing bound \cite{chiribella2011quantum} we may 
gain some insight of the upper bound for $r_1(\XX, \YY)$. If we assume that we 
are allowed to use only non-degenerated codes, then for perfectly correctable 
$\EE$ we have a bound of the form $\mathrm{rank}(J(\EE)) \le 
\frac{\dim(\YY)}{\dim(\XX)}$. 
In the next 
part of this section, we will improve the upper bound of $r_1(\XX,\YY)$ 
without putting any additional assumptions. We 
also will estimate the behavior of $r(\XX, 
\YY)$. In the particular case $\XX = \C^2$ and $\YY=\C^4$, we will also show 
that $r_1(\XX, \YY) < r(\XX, \YY)$.

Let us start with the following simple, but important properties, required to 
study $r(\XX,\YY)$.  We will notice, that for a constant Choi rank of the 
noise, it is easier to construct error-correcting scheme, if the dimension of 
$\YY$ is large.

\begin{lemma}\label{lem-weird-channels}
Let $\XX, \YY$ be Euclidean spaces. Define $Q \in \MM(\YY)$ such that $0 < 
Q\le \1_\YY$. Take $\EE \in \CC(\YY)$ and $\FF \in s\CC(\YY)$ given by $\FF(Y) 
= \EE(QYQ)$. Then, $\EE \in \xi(\XX, \YY)$ if and only if there exists a scheme 
${(\SS, \RR) \in s\CC(\XX,\YY) \times s\CC(\YY,\XX)}$ such that $0 \neq \RR \FF 
\SS \propto \II_\XX$.
\end{lemma}
Directly from Lemma~ \ref{lem-weird-channels} we receive the monotonicity 
of $r(\XX, \YY)$ w.r.t. the dimension of $\YY$. Let $\YY, \YY'$ be such 
Euclidean 
spaces that $\dim(\YY) \le \dim(\YY')$. Take $\EE = \kraus{(E_i)_i} \in 
\CC(\YY')$. There exist two projectors 
$\Pi_1, \Pi_2 \in \PP(\YY')$, such that $\mathrm{rank}(\Pi_1) = 
\mathrm{rank}(\Pi_2) = \dim(\YY)$ and for $ \FF = \kraus{(\Pi_2 E_i 
\Pi_1)_i}$ we have $\mathrm{rank}(\tr_1(J(\FF))) = \dim(\YY)$. Hence, if there 
exists a scheme $(\SS, \RR)$ such that $0 \neq \RR \FF \SS \propto \II_\XX$, 
then $\EE \in \xi(\XX, \YY')$. Eventually, we have
\begin{equation}
r(\XX,\YY) \le r(\XX,\YY').
\end{equation}

\subsection{Schur noise channels}
In this subsection, we restrict our attention to a particular family of 
noise channels whose Kraus operators are diagonal in the computational basis. 
In the 
literature, these channels are  referred to as Schur channels 
\cite[Theorem 4.19]{watrous2018theory}. We use them to study 
an upper bound for $r(\XX, \YY)$ and $r_1(\XX, \YY)$.
\begin{lemma}\label{lem-rank-bounds}
	Let $\XX$ and $\YY$ be Euclidean spaces such that $\dim(\YY) \ge 
	\dim(\XX)$. Then, there exists a Schur channel $\EE \in \CC(\YY)$ such that 
	$\mathrm{rank}(J(\EE)) = \left\lceil \frac{\dim(\YY)}{\dim(\XX) - 1} 
	\right\rceil$ 
	and $\EE \not\in \xi(\XX, 
	\YY)$. 
	Moreover, there exists a Schur channel $\FF \in \CC(\YY)$ such that 
	$\mathrm{rank}(J(\FF)) = \left\lceil\sqrt{\left\lceil 
	\frac{\dim(\YY)}{\dim(\XX) - 1} 
	\right\rceil}\right\rceil$ and $\FF \not\in \xi_1(\XX, 
	\YY)$. Especially, that implies
	\begin{equation}
	\begin{split}
		r(\XX,\YY) &< \frac{\dim(\YY)}{\dim(\XX) - 1},\\
		r_1(\XX,\YY) &< \sqrt{\frac{\dim(\YY)}{\dim(\XX)-1}}.
	\end{split}
	\end{equation}
\end{lemma}
The proof of Lemma~\ref{lem-rank-bounds} is presented in 
Appendix~\ref{proof-lem-rank-bounds}. The bounds obtained in 
Lemma~\ref{lem-rank-bounds} are asymptotically tight for Schur noise 
channels with $\dim(\YY) \to \infty$. To prove the tightness of the bound 
for perfectly correctable noise channels, we may use the construction provided 
in \cite{knill2000theory}. Hence, if we take a Schur channel $\EE = 
\kraus{(E_i)_i} \in \CC(\YY)$, such that $\mathrm{rank}(J(\EE)) 
\approx \sqrt{\frac{\dim(\YY)}{\dim(\XX)-1}}$, we obtain $\EE \in 
\xi_1(\XX, \YY)$. In the following proposition we will prove the tightness for 
probabilistically correctable Schur noise channels.

\begin{proposition}\label{prop-diag}
	Let $\XX$ and $\YY$ be Euclidean spaces and $\dim(\XX) \le \dim(\YY)$. For 
	any Schur channels $\EE \in \CC(\YY)$, such that $\mathrm{rank}(J(\EE)) < 
	\frac{\dim(\YY)}{\dim(\XX) - 1}$, it 
	holds $\EE \in \xi(\XX, \YY)$.
\end{proposition}
The proof of Proposition~\ref{prop-diag} is presented in 
Appendix~\ref{proof-prop-diag}. In the case of Schur channels we have a
clear separation between probabilistically and perfectly correctable noise 
channels. 

\subsection{From bi-linear to linear problem}

In general, the difficulty of finding error-correcting schemes $(\SS, 
\RR)$ comes from  bi-linearity of the problem Eq.~\eqref{eq-thm-1d}. However, 
there is a 
particular class of noise channels, for which we can easily rewrite the 
bi-linear 
problem as a linear one. In this subsection, we will 
focus our attention on noise channels $\EE \in 
\CC(\YY)$, such that $\mathrm{rank}(\EE(\1_\YY)) = \dim(\XX)$. Note, that this 
assumption implies $\dim(\XX) \mathrm{rank}(J(\EE)) \ge \dim(\YY)$. 

Let $\EE = \kraus{(E_i)_i}$ and let $\Pi$ be the projector on the image of 
$\EE(\1_\YY)$. Consider an associated channel $\FF = \kraus{(F_i)_i} =  
\kraus{(V_{\Pi}^\dagger 
E_i)_i} \in \CC(\YY, \XX)$, where $V_{\Pi} \in \MM(\XX, 
\YY)$ is an isometry operator with the image on the subspace defined by $\Pi$. 
It is clear that $\EE$ is probabilistically correctable for a given space 
$\XX$ if and only if there exists a scheme $(\SS, \RR)$, such that 
$0\neq\RR\FF\SS \propto \II_\XX$. Hence, according to Theorem 
\ref{thm-uqec-general} we need to find $S_* \in 
\MM(\XX, \YY), R_* \in \MM(\XX)$, such that $R_* F_i S_* = c_i \1_\XX$ and 
$c_{i_0} \neq 0$ for some $i_0$. Interestingly, we can combine together an 
action of $S_*, R_*$ as just the action of some pre-processing $S_*' \in 
\MM(\XX, \YY)$, that is 
\begin{equation}
\begin{split}
R_* F_i S_* = c_i \1_\XX \iff  F_i S_* R_* = c_i \1_\XX \iff F_i S_*' = c_i
\1_\XX.
\end{split}
\end{equation}
Therefore, we obtained a linear problem equivalent to Eq.~\eqref{eq-thm-1d}. In 
the following proposition we will investigate consequences of a such 
simplification.

\begin{proposition}\label{prop-bi-lin}
Let $\XX$ and $\YY$ be some Euclidean spaces and $\dim(\XX) \le \dim(\YY)$.
\begin{enumerate}[(A)]
\item If $\EE \in \CC(\YY)$ is a noise channel such that 
$\mathrm{rank}(\EE(\1_\YY)) 
= \dim(\XX)$ and $\mathrm{rank}(J(\EE)) < 
	\frac{\dim(\YY)\dim(\XX)}{\dim(\XX)^2 - 1}$, then 
	$\EE \in \xi(\XX, \YY)$.
\item There exists a noise channel $\EE \in 
	\CC(\YY)$ such that $\mathrm{rank}(\EE(\1_\YY)) = \dim(\XX)$ and 
	$\mathrm{rank}(J(\EE)) \ge 
	\frac{\dim(\YY)\dim(\XX)}{\dim(\XX)^2 - 1}$, for which 
	we have $\EE \not\in \xi(\XX, \YY)$.
\end{enumerate} 

\end{proposition}
The proof of Proposition~\ref{prop-bi-lin} is presented in 
Appendix~\ref{proof-prop-bi-lin}. Eventually, it is worth mentioning that 
the QEC procedure based on Knill-Laflamme conditions works well with this 
class of noise channels. Consider the situation $\dim(\XX)\mathrm{rank}(J(\EE)) 
= 
\dim(\YY)$. Then, if $\EE \in \CC(\YY)$ and $\mathrm{rank}(\EE(\1_\YY)) = 
\dim(\XX)$, it holds $\EE \in \xi_1(\XX, \YY)$. To see this, take the Kraus 
decomposition of $\EE = \kraus{(E_i)}$ and notice that operators $E_i$ are 
orthogonal pieces of some unitary operator.

\subsection{Correctable noise channels with bounded Choi rank}

In this subsection we will study the behavior of $r(\XX,\YY)$ and $r_1(\XX, 
\YY)$. We will state a lower and a upper bound for both 
quantities.

\begin{theorem}\label{thm-best-bounds}
	Let $\XX$ and $\YY$ be some Euclidean spaces such that $\dim(\YY) \ge 
	\dim(\XX)$. Then, we have
	\begin{equation}
	\begin{split}
		\left\lfloor \sqrt[4]{\frac{\dim(\YY)}{\dim(\XX)}} 
		\right\rfloor \le r_1(\XX,\YY) \le 		
		\left\lceil\sqrt{\frac{\dim(\YY)}{\dim(\XX)-1}}\right\rceil-1\le 
		r(\XX,\YY)
		 < 
		\frac{\dim(\YY)\dim(\XX)}{\dim(\XX)^2 - 1}.
	\end{split}
	\end{equation}
\end{theorem}
The proof of Theorem~\ref{thm-best-bounds} is presented in 
Appendix~\ref{proof-thm-best-bounds}. Unfortunately, according to this theorem, 
there is no clear separation of $r(\XX,\YY)$ and $r_1(\XX,\YY)$ for arbitrary 
$\XX$ and $\YY$. The improvement of these bounds will be investigated in the 
future. 

For now, we will calculate explicitly $r(\XX,\YY)$ and $r_1(\XX,\YY)$ for $\XX 
= \C^2$ and $\YY = \C^3, \C^4$.
\begin{proposition}\label{prop-qubit-rank}
	For all $\EE \in \CC(\C^4)$ satisfying $\mathrm{rank}(J(\EE)) \le 2$ we 
	have $\EE \in \xi(\C^2, \C^4)$.
\end{proposition}
The proof of Proposition~\ref{prop-qubit-rank} is presented in 
Appendix~\ref{proof-prop-qubit-rank}.  By 
using Theorem~\ref{thm-best-bounds} and Proposition~\ref{prop-qubit-rank} we 
get the following advantage of pQEC protocol for $\XX = \C^2$ and $\YY=\C^4$. 
\begin{corollary}\label{cor-exact-r}
	For $\XX = \C^2$ and $\YY = \C^4$ we have
	\begin{equation}
		r_1(\XX,\YY) < r(\XX,\YY).
	\end{equation}
In particular, it holds
\begin{equation}
\begin{array}{ccc}
	r_1(\C^2,\C^3) = 1 & \quad & r(\C^2,\C^3) = 1 \\
	r_1(\C^2,\C^4) = 1 & \quad & r(\C^2,\C^4) = 2
\end{array}
\end{equation}
\end{corollary}

\subsection{Random noise channels}

In the last subsection, we will show the advantage of pQEC procedure for 
randomly generated noise channels. We will follow the procedure of sampling 
quantum 
channels considered in \cite{bruzda2009random, 
nechita2018almost, kukulski2021generating}. 

Let $r \in \N$ and let $(G_i)_{i=1}^r \subset \MM(\YY)$ be a tuple of random 
and independent Ginibre matrices (matrices with independent and identically 
distributed entries drawn from standard complex normal distribution). Define $Q 
= \sum_{i=1}^r G_i^\dagger G_i$. We define a random channel $\EE_r \in 
\CC(\YY)$ given as
\begin{equation}\label{eq-random}
	\EE_r = \kraus{(G_i Q^{-1/2}))_{i=1}^r}.
\end{equation}
This sampling procedure induces the measure $\PP$ on $\CC(\YY)$ whose support 
is defined on $\{\EE \in \CC(\YY): \,\, \mathrm{rank}(J(\EE)) \le r\}$.

\begin{theorem}\label{thm-random}
	Let $\EE_r \in \CC(\YY)$ be a random quantum channel 
	defined according to Eq.~\eqref{eq-random}. Then, the following two 
	implications hold
	\begin{equation}
\begin{split}
r < \frac{\dim(\XX) \dim(\YY)}{\dim(\XX)^2 - 1} &\implies \PP\left( \EE_{r} \in 
		\xi(\XX, \YY)\right) = 1,\\
\PP\left( \EE_{r} \in \xi_1(\XX, \YY)\right) = 1 &\implies r < 
\sqrt{\frac{\dim(\YY)}{\dim(\XX)-1}}.
\end{split}
\end{equation}
\end{theorem}
The proof of Theorem~\ref{thm-random} is presented in 
Appendix~\ref{proof-thm-random}.

\section{Example of pQEC qubit code}\label{sec-simple}
Consider the following scenario. You have a task to transfer a given qubit 
state $\rho \in \DD(\C^2)$ through a quantum communication 
line represented by a noise channel $\EE \in \CC(\YY)$ of the form $\EE(Y) = 
\tr_2 \left(U(Y \otimes \proj{\psi}) U^\dagger\right)$, where $\proj{\psi} \in 
\DD(\C^2)$ and $U \in \UU(\YY \otimes \C^2)$. At this point a natural question 
arises. What is the minimal size of the communication line $\dim(\YY)$, which 
is large enough to recover the state $\rho$ with the pQEC procedure?

To answer this question, observe that the channel $\EE$ satisfies 
$\mathrm{rank}(J(\EE)) \le 2$. In Proposition~\ref{prop-qubit-rank} we noticed 
that such channels are probabilistically correctable for a given input space 
$\C^2$, if $\dim(\YY) = 4$ (in fact, from monotonicity for $\dim(\YY) \ge 4$).
Therefore, to correctly transfer a qubit state through $\EE$, we may define 
an error-correcting scheme with only two physical qubits. 

We provide the following pQEC procedure based on 
Proposition~\ref{prop-qubit-rank}. 

\begin{algorithm}[H]
\SetAlgoLined
\KwIn{$\EE \in \CC(\C^4)$ such that 
$\mathrm{rank}(J(\EE))\le 2$.}
\KwOut{pQEC procedure with success probability $p>0$.}
  \BlankLine\BlankLine

\nl Let $\EE = \kraus{(E_0, E_1)}$.\\
\nl Define $S_* \in \MM(\C^2, \C^4)$ and $R_* \in \MM(\C^4, \C^2)$, 
such that $R_* E_0 S_* \propto \1_{\C^2}$, $R_* E_1 S_* \propto \1_{\C^2}$ and 
$R_* E_0 S_* \neq 0 \vee R_* E_1 S_* \neq 0$ according to 
Appendix~\ref{proof-prop-qubit-rank}.\\
\nl Define \begin{equation*}
\begin{split}
Q &= S_*^\dagger S_*,\\
S &=  S_* Q^{-0.5},\\
R &= \frac{Q^{0.5} R_*}{\|Q^{0.5} R_* \|_\infty}.
\end{split}
\end{equation*}\\
\nl Calculate $p \in (0,1]$, such that 
$R\left(\EE\left(SXS^\dagger\right)\right)R^\dagger = pX$ for any $X \in 
\MM(\C^2).$\\
\nl Define $U_S \in \UU(\C^4)$ which satisfies 
$U_S(\1_{\C^2} 
\otimes \ket{0}) = S$.\\
\nl Let $R = \sigma_1 \ketbra{z_1}{t_1} + \sigma_2 \ketbra{z_2}{t_2}$ be 
	the SVD of $R$. Define $U_R \in \UU(\C^4)$ which 
	satisfies
	\begin{equation*}
		\begin{split}
			U_R \ket{t_1} &= \ket{0,0},\\
			U_R \ket{t_2} &= \ket{1, 0}.
		\end{split}
	\end{equation*}\\
\nl Define $R' = R U_R^\dagger (\1_{\C^2} \otimes \ket{0})$.\\
\nl Define $V_R \in \UU(\C^4)$ which satisfies $(\1_{\C^2} 
\otimes \bra{0})V_R (\1_{\C^2} \otimes \ket{0}) = R'.$\\
\nl Run the QEC procedure presented in Figure~\ref{scheme} for $\ket{\psi}, 
U_S, U_R, V_R$.\\
\nl Let $\sigma_{\mathrm{exp}}$ be the output state of the procedure 
	presented in Figure \ref{scheme}. Use the post-processing of the 
	measurements' output $(i,j)$ according to the following table:
	
	\begin{equation*}
		\begin{array}{|c|c|c|c|c|}
			\hline
			\mathrm{\textbf{Labels}} & 
			\mathrm{\textbf{Probability}}&\mathrm{\textbf{Status}} & 
			\mathrm{\textbf{Action}} & 
			\mathrm{\textbf{Result}} \\
			\hline
			(i, j) = (0,0) & p & \mbox{QEC procedure succeeded } & \mbox{Accept 
			} 
			\sigma_{\mathrm{exp}} & 
			\sigma_{\mathrm{exp}} = 
			\proj{\psi} \\
			(i, j) \neq (0,0) & 1-p & \mbox{QEC procedure failed } & 
			\mbox{Reject } 
			\sigma_{\mathrm{exp}} & \sigma_{\mathrm{exp}}\,\, ? \,\, 
			\proj{\psi} \\\hline
		\end{array}
	\end{equation*}
\caption{Probabilistic QEC qubit code}\label{alg}
\end{algorithm}

\begin{figure}[h!]
	\includegraphics[width=0.9\linewidth]{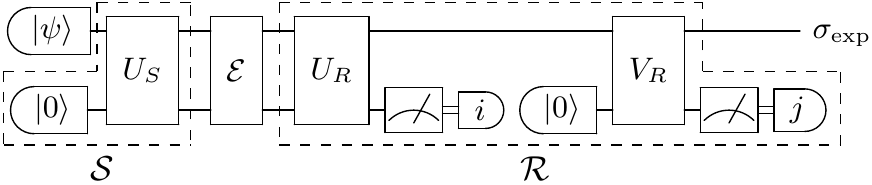}
	\caption{The circuit representing the pQEC procedure. We have access to two 
	physical qubits. The first qubit is in the state $\ket{\psi}$. 
		This state will be encoded. The second state we set
		$\ket{0}$. We implement two-qubit, encoding unitary operator $U_S$. 
		Then, the encoded state $U_S \left(\ket{\psi} \otimes \ket{0} 
		\right)$ is affected by the noise channel $\EE$. After that, we start 
		the 
		decoding 
		procedure. We implement two-qubit unitary 
		rotation $U_R$. We measure the second qubit in the standard basis 
		and 
		obtain a 
		classical label $i \in \{0,1\}$. We prepare a third qubit in the 
		state $\ket{0}$ and implement two qubit unitary rotation 
		$V_R$. 
		We measure the third qubit in the standard basis and obtain 
		a classical label $j \in \{0,1\}$. If $(i,j)=(0,0)$ we accept the 
		output state, otherwise, we reject it.\label{scheme}}
\end{figure}
	
\section{Generalization of pQEC procedure}\label{sec-model}

Let us denote by $\Upsilon$ an arbitrary family of noise channels, that is 
$\Upsilon 
\subset \CC(\YY)$. In this section, we ask if there exists error-correcting 
scheme $(\SS,\RR)$, such that all noise channels $\EE \in \Upsilon$ we have 
$\RR \EE 
\SS 
= p_\EE \II_\XX$, for some $p_\EE \ge 0$. Note, that $p_\EE$ may differ for 
different noise channels 
$\EE$, hence, we shall introduce a quantity to ``globally'' control the 
effectiveness of $(\SS, \RR)$. We propose the following approach.

Let $\mu$ be some probability measure defined on the set $\Upsilon$. We assume 
that 
noise channels $\EE \in \Upsilon$ are probed according to $\mu$. The scheme 
$(\SS, 
\RR)$ 
will be a valid error-correcting scheme for $\Upsilon$ and $\mu$ if 
in average, the probability of successful error correction is non zero, that is 
\begin{equation}
	\int_\Upsilon p_\EE \mu(d\EE) > 0.
\end{equation}
Without loss of the generality we may assume that $\Upsilon$ is convex. 
Additionally, we assume that the support of $\mu$ is equal to $\Upsilon$. 
Usually, 
we can take $\mu$ as the flat measure, representing the maximal uncertainty in 
the process of probing random noise channels $\EE$ from $\Upsilon$. Let us 
define 
the 
average noise channel of $\Upsilon$ with respect to $\mu$
\begin{equation}
\bar \EE = \int_{\Upsilon} \EE \mu(d\EE).
\end{equation} 
We will show that we can correct \emph{all} 
noise channels from the family $\Upsilon$, whenever $\bar \EE$ is 
probabilistically correctable for $\XX$. We put this statement as the following 
proposition.

\begin{proposition}\label{prop-noise-model}
Let $\Upsilon \subset \CC(\YY)$ be a nonempty and convex family of noise 
channels. Define $\mu$ to be a probability measure defined on $\Upsilon$ and 
assume that the support of $\mu$ is equal to $\Upsilon$. Let $\bar \EE = 
\int_{\Upsilon} \EE \mu(d\EE) \in \CC(\YY)$ and fix 
$(\SS, \RR) \in s\CC(\XX,\YY) \times s\CC(\YY,\XX)$. The following conditions 
are equivalent:
\begin{enumerate}[(A)]
\item For each $\EE \in \Upsilon$ there exists $p_\EE \ge 0$ such that $\RR 
\EE \SS = p_\EE \II_\XX$ and $	\int_{\Upsilon} p_\EE \mu(d\EE) > 0.$
\item It holds that $0 \neq \RR \bar \EE \SS \propto \II_\XX$.
\end{enumerate}
\end{proposition}
The proof of Proposition~\ref{prop-noise-model} is presented in 
Appendix~\ref{proof-prop-noise-model}.

\section{Discussion}
In this work, we analyzed pQEC procedure for a general noise model. We 
established 
the conditions to check if a given noise channel is probabilistically 
correctable. Moreover, we showed that mixed state encoding should be taken into 
account when maximizing the probability of successful error correction. 
Finally, we pointed the advantage of the probabilistic error-correcting 
procedure over the deterministic one. We saw a clear separation especially for 
a correction of Schur noise channels and random noise channels. We obtained the 
maximum value of Choi rank of probabilistically correctable noise channels. We 
also provide a method how to probabilistically correct noise channels with 
bounded Choi rank. 

There are many directions for further study that still remain to be 
explored. It would be interesting to strengthen Theorem~\ref{thm-best-bounds} 
and show the separation between $r(\XX,\YY)$ and $r_1(\XX,\YY)$ by improving 
the proposed proof technique in Appendix~\ref{proof-thm-best-bounds}. We 
obtained such separation for $\XX = \C^2$ and $\YY = \C^4$ in 
Corollary~\ref{cor-exact-r}. Another 
promising direction 
is to propose tools for 
the numerical analysis of pQEC protocols, based on 
Theorem~\ref{thm-uqec-general}. Such tools would help us estimate the value of 
$r(\XX, \YY)$ and gain an insight into probabilistically correctable noises 
that require mixed state encoding. Last but not least, we would like to 
calculate the worst-case probability of successful error correction for a given 
noise intensity $r \le r(\XX, \YY)$. For example, as we showed in 
Proposition~\ref{prop-qubit-rank}, the errors 
caused by a unitary interaction with an auxiliary qubit system ($r=2$), can be 
corrected by using only two physical qubits ($\dim(\YY)=4$). We can ask, how 
many times in average the procedure presented in Algorithm~\ref{alg} needs to 
be repeated.

\section*{Acknowledgments}

This work was supported by the project ``Near-term Quantum Computers: 
challenges, optimal implementations and applications'' under
Grant Number POIR.04.04.00-00-17C1/18-00, which is carried out within the 
Team-Net programme of the
Foundation for Polish Science co-financed by the European Union under the 
European Regional
Development Fund.

\bibliographystyle{ieeetr}
\bibliography{unambiguous.bib}
\newpage
\appendix

\section{}\label{sec-app}
\subsection{Constant probability of successful error 
correction}\label{proof-lem-constant-p}
\begin{lemma}\label{lem-constant-p}
	Let $\EE \in \CC(\YY)$, $\SS \in s\CC(\XX,\YY)$ and $\RR \in 
	s\CC(\YY,\XX)$. If for any pure state $\proj{\psi} \in \DD(\XX)$ it holds
	$\RR \EE \SS (\proj{\psi}) \propto \proj{\psi}$, then there exists $p \in 
	[0,1]$ such 
	that $\RR \EE \SS = p \II_\XX$.
\end{lemma}
\begin{proof}
	Let $\LL = \RR \EE \SS$ and for any unitary operator $U \in \UU(\XX)$ and 
	$i 
	= 0,\ldots,\dim(\XX)-1$ define $p_{U,i} \in [0,1]$ by 
	$\LL(U\proj{i}U^\dagger) = p_{U,i} 
	U\proj{i}U^\dagger$. We have $\LL(\1_\XX) = U 
	\left(\sum_i p_{U,i} 
	\proj{i} \right) U^\dagger$ for any $U$ and hence, there exists $p \in 
	[0,1]$ such that 
	$\LL(\1_\XX) = p \1_\XX$. That means, $ p_{U, i} = p$ for any $U$ and $i$, 
	so $\LL(\proj{\psi}) 
	= p\proj{\psi}$ 
	for any 
	$\proj{\psi} \in \DD(\XX)$. We obtain the thesis by 
	noting that $\text{span}_{\C}(\proj{\psi}) = \MM(\XX)$.
\end{proof}

\subsection{Proof of 
Theorem~\ref{thm-uqec-general}}\label{proof-thm-uqec-general}
\begin{non}{Theorem~\ref{thm-uqec-general}}
	Let $\EE = \kraus{(E_i)_i} \in \CC(\YY)$. The following conditions are
	equivalent:
	\begin{enumerate}[(A)]
		\item There exist error-correcting scheme $(\SS, \RR) 
		\in s\CC(\XX,\YY) \times s\CC(\YY,\XX)$ and $p>0$ such that
		\begin{equation}
			\RR \EE \SS  = p \II_\XX.
		\end{equation}
		\item There exist $S
		= \kraus{(S_k)_k} \in s\CC(\XX,\YY)$ and $R \in \PP(\YY)$, such that $R 
		\leq \1_\YY$, for which it holds
		\begin{equation}
			\kraus{ \left( \sqrt{R} E_i S_k \right)_{i,k} } = \kraus{(A_i)_i}: 
			\quad A_i \not = 0, A_j^\dagger A_i \propto
			\delta_{ij} \1_\XX.
		\end{equation}
		\item There exist $\SS =
		\kraus{(S_k)_k} \in s\CC(\XX,\YY)$, $R \in \PP(\YY)$, such that $R \leq 
		\1_\YY$ and a matrix $M= [M_{jl,ik}]_{jl,ik} 
		\not=0$, for which it
		holds
		\begin{equation}
			\forall_{i,j,k,l} \quad S_l^\dagger E_j^\dagger R E_i S_k = M_{jl, 
				ik} \1_\XX.			
		\end{equation}
		\item There exist $S_* \in \MM(\XX,\YY)$ and $R_* \in \MM(\YY, 
		\XX) $ such that
		\begin{equation}
			\forall_i \quad R_* E_i S_* \propto \1_\XX
		\end{equation}
		and there exists $i_0$, for which it holds $R_* E_{i_0} S_* \neq 0$.
	\end{enumerate}
	
	Moreover, if point $(A)$ holds for $\SS = \kraus{(S_k)_k}$ and $\RR = 
	\kraus{(R_l)_l}$, then $R \in \PP(\YY)$ from points $(B)$ and $(C)$ 
	can be chosen to satisfy $	R = \sum_l R_l^\dagger R_l.$ It also holds that
	$R_l E_i S_k \propto \1_\XX$  for any $i,k,l$.
\end{non}
\begin{proof}
	In order to show that $(A) \iff (B) \iff (C)$, in all implications 
	presented below, we will use the same encoding $\SS = 
	\kraus{(S_k)_k} \in s\CC(\XX, \YY)$. Hence, to simplify the proof, we 
	introduce the notation of $\FF 
	\coloneqq \EE \SS$ given in the form $\FF = \kraus{(F_i)_i}$.\\
	
	$(B) \implies (A)$\\
	Let us define $\alpha_i>0$ to satisfy $A_i^\dagger A_i = \alpha_i \1_\XX$
	and operation $\RR: \MM(\YY) \mapsto \MM(\XX)$ given by
	\begin{equation}
		\begin{split}
			\RR &= \kraus{\left( \alpha_i^{-1/2} A_i^\dagger \sqrt{R} 
			\right)_i}.
		\end{split}
	\end{equation}
	We will check that $\RR$ is a subchannel. First, from the definition of 
	$\RR$, it follows that $\RR$ is completely positive. Second, from the
	assumption $(B)$, operators $\alpha_i^{-1} A_i A_i^\dagger \in \PP(\YY)$ are
	orthogonal projectors and hence
	\begin{equation}
		\sum_i \alpha_i^{-1} \sqrt{R} A_i A_i^\dagger \sqrt{R}=\sqrt{R} \left( 
		\sum_i
		\alpha_i^{-1}  A_i A_i^\dagger \right) \sqrt{R} \leq R \leq \1_\YY.
	\end{equation}
	It means that $\RR \in s\CC(\YY,\XX)$. Finally, it holds
	\begin{equation}
		\RR \FF  = \kraus{\left(\alpha_j^{-1/2}A_j^\dagger \sqrt{R} F_i 
		\right)_{i,j}} =
		\kraus{\left(\alpha_j^{-1/2}A_j^\dagger 
		A_i\right)_{i,j}}=\kraus{(\alpha_i^{1/2} \1_\XX)_i}  = 
		p 
		\II_\XX,
	\end{equation}
	where we introduced $p \coloneqq \sum_i \alpha_i > 0$. \\
	
	$(A) \implies (B)$\\
	Let $\RR = \kraus{(R_k)_k}$ and take $R = \sum_k R_k^\dagger R_k \in 
	\PP(\YY)$.
	From $\RR \in s\CC(\YY,\XX)$ it follows $R \le \1_\YY$. Define $\Pi_R$ to 
	be the projector on the support of $R$. One can show that $R_k \Pi_R =
	R_k$ for each $k$. We define $\widetilde{\RR} = 
	\kraus{\left(\widetilde{R_k}\right)_k}$, where $ 
	\widetilde{R_k} \coloneqq R_k \sqrt{R}^{-1}$. 
	From the definition of $\widetilde{\RR}$ we have $ \sum_k 
	\widetilde{R_k}^\dagger 
	\widetilde{R_k} = \sqrt{R}^{-1} R\sqrt{R}^{-1} = \Pi_R.$ Using the 
	assumption $(A)$ we get $p \II_\XX = \RR \FF = \widetilde{\RR} \circ 
	\kraus{(\sqrt{R} 
	F_i)_i}$. As we have $p>0$, it follows that $\kraus{(\sqrt{R} F_i)_i } \not 
	= 
	0$. 
	Hence, 
	there exists a canonical decomposition
	\begin{equation}
		\kraus{(\sqrt{R} F_i)_i } = \kraus{(A_i)_i}: \quad A_i \not=0, 
		\tr(A_j^\dagger 
		A_i) = 0 \text{ 
		for }
		i \neq j.
	\end{equation}
	From the relationship between Kraus representations, it follows that $A_i$ 
	satisfy $\Pi_R A_i = 
	A_i$. Then, by
	Choi-Jamio{\l}kowski isomorphism we have
	\begin{equation}
		p \proj{\1_\XX} = (\RR \FF \otimes \II_\XX)(\proj{\1_\XX})=\sum_i 
		(\widetilde{\RR} \otimes \II_\XX)(\proj{A_i}).
	\end{equation}
	Therefore, from the extremality of the point $\proj{\1_\XX}$ in $\PP(\XX 
	\otimes \XX)$ we have
	\begin{equation}
		\widetilde{R_k} A_i \propto \1_\XX
	\end{equation}
for any $i,k$. On the one hand we get
	\begin{equation}
		\sum_k A_j^\dagger \widetilde{R_k}^\dagger \widetilde{R_k} A_i 
		\propto
		\1_\XX
	\end{equation}
	and on the other hand
	\begin{equation}
		\sum_k A_j^\dagger \widetilde{R_k}^\dagger \widetilde{R_k}
		A_i=A_j^\dagger \Pi_R A_i = A_j^\dagger A_i.
	\end{equation}
	The above conditions provide that $A_j^\dagger A_i = c_{ji} 
	\1_\XX$, for some $c_{ji} \in \C$. Then, for $i \neq j$ we have $0=\tr 
	(A_j^\dagger
	A_i)=c_{ji} \dim(\XX)$ and eventually $A_j^\dagger A_i = 0$.\\
	
	$(B) \implies (C)$\\
	Let us define $\alpha_k>0$ to satisfy $A_k^\dagger A_k = \alpha_k \1_\XX$.
	From the relationship between Kraus decompositions $\kraus{(\sqrt{R} F_i)_i 
	}$ and $\kraus{(A_i)_i}$, there	exists isometry operator $U$, such that
	\begin{equation}
		\sqrt{R} F_i = \sum_k U_{ik} A_k.
	\end{equation}
	Therefore, it holds
	\begin{equation}
		F_j^\dagger R F_i = \sum_{k,k'} U_{ik}
		\overline{U_{jk'}}A_{k'}^\dagger A_k = \sum_k U_{ik} \overline{U_{jk}}
		\alpha_k \1_\XX.
	\end{equation}
	Let us define a matrix $M = [M_{j,i}]_{j,i}$ where $M_{j,i} = \sum_k U_{ik} 
	\overline{U_{jk}}
	\alpha_k$.	Note, that
	\begin{equation}
		\tr(M) = \sum_{i,k} |U_{ik}|^2 \alpha_k = \sum_k \alpha_k > 0.
	\end{equation}

	$(C) \implies (B)$\\
	Let us define a operator $F=\sum_i \bra{i} \otimes F_i$. From the
	assumption $(C)$ it follows
	\begin{equation}
		F^\dagger R F  = M \otimes
		\1_\XX.
	\end{equation}
	That implies $M \ge 0$. Take the spectral decomposition $M =
	U^\dagger D U$ and define
	\begin{equation}
		A_i = \sum_k \overline{U_{ik}} \sqrt{R} F_k.
	\end{equation}
	Observe that $\kraus{(\sqrt{R} F_i)_i }  = \kraus{(A_i)_i}$. We 
	obtain
	\begin{equation}
		A_j^\dagger A_i = \sum_{k,k'}  \overline{U_{ik}} U_{jk'}  
		F_{k'}^\dagger R
		F_k  = \sum_{k,k'} \overline{U_{ik}} U_{jk'} M_{k'k} \1_\XX = 
		D_{ji}
		\1_\XX.
	\end{equation}
	This is equivalent to $A_j^\dagger A_i \propto \delta_{ij} \1_\XX$. Finally,
	$A_i \neq 0$ if and only if $D_{ii}>0$ and by the fact $M \neq 0$ we
	conclude the set $\{A_i: A_i \neq 0\}$ is not empty.\\
	
	$(A) \implies (D)$\\
	Let $\SS = \kraus{(S_k)_k}$ and $\RR = \kraus{(R_l)_l}$. Using the 
	Choi-Jamio{\l}kowski isomorphism we get
	\begin{equation}
		\left(\RR\EE\SS \otimes \II_\XX  \right)\left(\proj{\1_\XX}\right) = 
		\sum_{l,i,k} \proj{R_l E_i S_k} = p \proj{\1_\XX}.
	\end{equation}
	Therefore, from the extremality of the point $\proj{\1_\XX}$ in $\PP(\XX 
	\otimes \XX)$ we obtain $R_l E_i S_k \propto \1_\XX$. There exist $l_0, 
	i_0, k_0$ such that $R_{l_0} E_{i_0} S_{k_0} \neq 0$. We can take $S_* = 
	S_{k_0}$ and $R_* = R_{l_0}$.\\
	
	$(D) \implies (A)$\\
	There exist $q_0, q_1 > 0$ for which $\SS 
	\coloneqq q_0\kraus{(S_*)} \in s\CC(\XX, \YY)$ and $\RR \coloneqq 
	q_1\kraus{(R_*)} \in s\CC(\YY, \XX)$. One may note that $0 \neq \RR \EE 
	\SS\propto \II_\XX$.\\
	
	$(*)$\\
	Assume that $\SS = \kraus{(S_k)_k}$, $\RR = 
	\kraus{(R_l)_l}$ and it holds $(A)$. From the proof of 
	implications $(A) \implies (B)$ and $(B) \implies (C)$ it follows that 
	$R$ can be chosen as $R = \sum_l R_l^\dagger R_l$. The relation $R_lE_iS_k 
	\propto \1_\XX$ was proven in $(A) \implies (D)$. 	
\end{proof}

\subsection{Proof of 
	Proposition~\ref{prop-realization}}\label{proof-prop-realization}	
\begin{non}{Proposition~\ref{prop-realization}}
	For a given channel $\EE \in \CC(\YY)$, let us fix an error-correcting 
	scheme $(\SS, \RR) \in s\CC(\XX,\YY) \times s\CC(\YY,\XX)$ such that $
	\RR \EE \SS = p \II_\XX$, for some $p > 0.$ Then, the following holds:
	\begin{enumerate}[(A)]
		\item There exist $\widetilde \SS \in \CC(\XX,\YY)$ and $\widetilde \RR 
		\in 
		s\CC(\YY,\XX)$ such that $\widetilde \RR \EE \widetilde \SS = p 
		\II_\XX.$
		\item If $\RR \in \CC(\YY,\XX)$, then there exists $\widetilde \SS = 
		\kraus{(\widetilde S)} 
		\in \CC(\XX,\YY)$ such that $\RR \EE \widetilde \SS = \II_\XX.$
		\item If $p=1$, then there exist $\widetilde \SS = \kraus{(\widetilde 
		S)} 
		\in \CC(\XX,\YY)$ and $\widetilde \RR \in \CC(\YY,\XX)$ such that $
		\widetilde \RR \EE \widetilde \SS = \II_\XX.$
	\end{enumerate}
\end{non}
\begin{proof}
	$(A)$\\
	Let $\SS = \kraus{(S_k)_k}$ and $S = \sum_k S_k^\dagger S_k \le \1_\XX$. 
	Using 
	Theorem~\ref{thm-uqec-general} one can show that there exists $k_0$ for 
	which $\mathrm{rank}(S_{k_0})=\dim(\XX)$. Hence, $S$ is invertible.
	Define $\widetilde{\SS} \in \CC(\XX,\YY), \widetilde{\RR} \in 
	s\CC(\YY,\XX)$ given 
	by the equations
	\begin{equation}
		\begin{split}
			\widetilde{\SS}(X) &= \SS \left( S^{-1/2} X S^{-1/2} \right),\\
			\widetilde{\RR}(Y) &= S^{1/2} \RR(Y) S^{1/2}.
		\end{split}
	\end{equation}
	We obtain $\widetilde{\RR} \EE \widetilde{\SS}(X) = S^{1/2} (\RR \EE 
	\SS) \left(S^{-1/2} X 
			S^{-1/2}\right) S^{1/2}= pX.$\\
	
	$(B)$\\
	Let $\SS = \kraus{ (S_k)_k }$ and define $\SS_k(X) = S_kXS_k^\dagger$. From 
	Theorem~\ref{thm-uqec-general} there exists $k_0$ such that $\RR \EE 
	\SS_{k_0} = p_{k_0} \II_\XX$, for some $p_{k_0} > 0$. For any $\proj{\psi} 
	\in \DD(\XX)$ it holds 
	then
	\begin{equation}
		p_{k_0} = \tr \left( \RR \EE \SS_{k_0} (\proj{\psi})\right) = \tr 
		\left( 
		\SS_{k_0} 
		(\proj{\psi})\right) = \bra{\psi} S_{k_0}^\dagger S_{k_0} \ket{\psi}.
	\end{equation}
	Hence, we get $S_{k_0}^\dagger S_{k_0} = p_{k_0} \1_\XX$. Define 
	$\widetilde 
	\SS = 
	\frac{1}{p_{k_0}} \SS_{k_0} \in \CC(\XX,\YY)$ and note that $\RR \EE 
	\widetilde 
	\SS 
	= \II_\XX$.\\
	
	$(C)$\\
	Let $\SS = \kraus{(S_k)_k}$ and $\RR = \kraus{(R_k)_k}$. For any 
	$\proj{\psi} 
	\in 
	\DD(\XX)$ we have
	\begin{equation}
		1 = \tr(\RR\EE\SS(\proj{\psi})) \le \tr(\SS(\proj{\psi})) \le 1.
	\end{equation}
	Therefore, for any $\proj{\psi} \in \DD(\XX)$ we get $ \bra{\psi} 
	\left(\sum_k 
	S_k^\dagger S_k \right) \ket{\psi} = 1$, which implies $\SS \in \CC(\XX, 
	\YY)$. Let $R = \sum_k R_k^\dagger R_k \le \1_\YY$. Then, it holds $
		\tr \left( (\1_\YY - R) \EE \SS (X) \right) = 0.$
	Define $\widetilde \RR \in \CC(\YY,\XX)$ by the equation
	\begin{equation}
		\widetilde \RR (Y) = \RR(Y) + \tr \left( (\1_\YY - R) Y \right) 
		\rho_{\XX}^*.
	\end{equation}
	Observe that $\widetilde \RR \EE \SS = \II_\XX$. The rest of the proof 
	follows from $(B)$.
\end{proof}

\subsection{Proof of Lemma~\ref{lem-example-gen}}\label{proof-lem-example-gen}
\begin{non}{Lemma~\ref{lem-example-gen}}
	Let $R \in \PP(\C^4)$ and $R \le \1_{\C^4}$. Define $\Pi_R$ as a projector 
	on 
	the 
	support of $R$. For $\EE_R$ defined as 
	\begin{equation}
		\EE_R(Y) = \proj{0} \otimes \tr_1\left( \sqrt{R} Y \sqrt{R}\right) + 
		\proj{1} \otimes \tr\left([\1_{\C^4} - R] Y\right) \rho_2^*
	\end{equation} we have the following simplified form of the 
	maximization problem $p_0(R)$:
	\begin{equation}\label{eq-param-noise-lem-app}
		p_0(R) = \max\left\{ \tr(P): \,\, P \in \PP(\C^2), \tr_1\left(R^{-1} (P 
		\otimes \1_{\C^2})\right) \le \1_{\C^2}, 
		\,\, 
		\forall_{X \in \MM(\C^2)} \, \,\Pi_R (P \otimes X) \Pi_R = P \otimes X 
		\right\}.
	\end{equation}
	An optimal scheme $(\SS, \RR)$ which achieves the probability $p_0(R)$, 
	that is 
	$\RR \EE_R \SS = p_0(R) \II_{\C^2}$, can be taken as
	\begin{equation}
		\begin{split}
			\SS(X) &= \sqrt{R}^{-1} (P \otimes X) \sqrt{R}^{-1},\\
			\RR(Y) &= \tr_1\left(Y\left(\proj{0} \otimes \1_{\C^2}\right) 
			\right),
		\end{split}
	\end{equation}	
	where $P$ is an argument maximizing $p_0(R)$ in 
	Eq.~\eqref{eq-param-noise-lem-app}. 
	Moreover, if there exists another optimal scheme $(\widetilde\SS, 
	\widetilde\RR)$, that is $\widetilde\RR \EE_R \widetilde \SS = p_0(R) 
	\II_{\C^2}$, 
	then $\mathrm{rank}(J(\SS)) \le \mathrm{rank}(J(\widetilde{\SS}))$.
\end{non}
\begin{proof}
Let us investigate the form of an optimal scheme $(\SS, \RR)$ that maximize the 
probability $p$ of successful error correction, $\RR \EE_R \SS = p \II_{\C^2}$. 
First, one can note that $\RR$ must be of the form $\RR(A \otimes B) = 
\tr\left(A \proj{0}\right) \widetilde \RR (B)$, where $\widetilde{\RR} = 
\kraus{(\widetilde{R_k})_k} \in s\CC(\C^2)$. Let us introduce an operation $\FF 
= \kraus{(F_i)_i} \in s\CC(\C^2)$ given by $\FF(X) = 
\tr_1\left(\sqrt{R}\SS(X)\sqrt{R}\right)$. We obtain $p \II_{\C^2} = \RR 
\EE_R \SS = \widetilde \RR \FF$. From Theorem~\ref{thm-uqec-general} we have 
$\widetilde R_k F_i \propto \1_{\C^2}$ and there are $k_0, i_0$ such that 
$\widetilde R_{k_0} F_{i_0} \neq 0$. Hence, for each $k$ we have $\widetilde 
R_k \propto F_{i_0}^{-1}$. That implies the operation $\widetilde \RR$ can be 
written as $\widetilde\RR(X) = \widetilde R X \widetilde R^\dagger$.
Now, consider another scheme $(\SS', \RR')$, where 
$\RR'(A \otimes B) = \tr\left(A \proj{0}\right) B$ and $\SS'(X) = 
S\left(\widetilde R X \widetilde R^\dagger\right) \in s\CC(\C^2, 
\C^4)$. We get 
\begin{equation}
	\RR'\EE_R\SS'(X) = \tr_1\left(\sqrt{R}\SS\left( \widetilde R X  
	\widetilde R^\dagger \right)\sqrt{R}\right) = \FF \left( \widetilde R X  
	\widetilde R^\dagger \right) = \left(\widetilde R\right)^{-1} \left( 
	\widetilde 
	\RR \FF 
	\left( \widetilde R X  
	\widetilde R^\dagger \right)\right) \left(\widetilde R^\dagger\right)^{-1} 
	= p X.
\end{equation}
Therefore, the scheme $(\SS', \RR')$ is also optimal and 
$\mathrm{rank}(J(\SS'))
\le \mathrm{rank}(J(\SS))$. 

To sum up, from now, we will consider the optimal scheme $(\SS, \RR)$, where 
$\RR(Y) = \tr_1\left(Y\left(\proj{0} \otimes \1_{\C^2}\right) \right)$. The 
equation 
$\RR \EE_R \SS = p \II_{\C^2}$ can be rewritten as 
\begin{equation}\label{eq-example-eq}
	\tr_1\left(\sqrt{R}\SS(X)\sqrt{R}\right) = pX,
\end{equation}
for any $X \in \MM(\C^2)$. According to Theorem~\ref{thm-uqec-general} we have
$\sqrt{R}\SS(X)\sqrt{R} = \sum_i A_i X A_i^\dagger,$ where $A_j^\dagger A_i 
\propto \delta_{ij} \1_{\C^2}.$ Using Theorem~\ref{thm-uqec-general} to the 
equation 
$\tr_1(\sum_i A_i X A_i^\dagger) = pX$ we obtain that $A_i = \ket{v_i} \otimes 
\1_{\C^2}$ for some orthogonal vectors $\ket{v_i} \in \C^2$. Let $P = \sum_i 
\proj{v_i}$. We get
$\sqrt{R}\SS(X)\sqrt{R} = P \otimes X.$ Without loss of 
the generality we may consider $\SS$ such that $\Pi_R \SS(X)\Pi_R = 
\SS(X)$ (one can note that $\mathrm{rank}(J(\SS))$ will not increase).
Hence, the equation $\sqrt{R} \SS(X) \sqrt{R} = P \otimes X$ implies 
$\SS(X) = \sqrt{R}^{-1} (P \otimes X) \sqrt{R}^{-1}.$
The condition $\sqrt{R} \SS(X) \sqrt{R} = P \otimes X$ becomes 
now equivalent with $\Pi_R (P \otimes X) \Pi_R = P \otimes X$ and the condition 
$\SS \in s\CC(\C^2, \C^4)$ is then equivalent with $\tr_1\left(R^{-1} (P 
\otimes \1_{\C^2})\right) \le \1_{\C^2}$. Therefore, basing on 
Eq.~\eqref{eq-example-eq} we can express the probability $p_0(R)$ as:
\begin{equation}
\begin{split}
p_0(R) &= \max\left\{p: \,\, \RR \EE_R \SS = p \II_{\C^2}, \,\, (\SS, \RR) 
\in s\CC(\C^2, \C^4) \times s\CC(\C^4,\C^2)\right\}\\
&= \max\left\{p: \,\,  	\forall_{X \in \MM(\C^2)} \, 
\,\tr_1\left(\sqrt{R}\SS(X)\sqrt{R}\right) = pX, \,\, \SS 
\in s\CC(\C^2, \C^4)\right\}\\
&= \max\left\{ \tr(P): \,\, P \in \PP(\C^2), \tr_1\left(R^{-1} (P 
	\otimes \1_{\C^2})\right) \le \1_{\C^2}, 
	\,\, 
	\forall_{X \in \MM(\C^2)} \, \,\Pi_R (P \otimes X) \Pi_R = P \otimes X 
	\right\}.
\end{split}
\end{equation}
\end{proof}

\subsection{Proof of 
Corollary~\ref{cor-example-rank-123}}\label{proof-cor-example-rank-123}
\begin{non}{Corollary~\ref{cor-example-rank-123}}
	Let us take $R \in \PP(\C^4)$ such that $R \le \1_{\C^4}$ and 
	$\mathrm{rank}(R) 
	< 
	4$. Define $\Pi_R$ as a projector on 
	the 
	support of $R$. For the noise channel defined as 
\begin{equation}
	\EE_R(Y) = \proj{0} \otimes \tr_1\left( \sqrt{R} Y \sqrt{R}\right) + 
	\proj{1} \otimes \tr\left([\1_{\C^4} - R] Y\right) \rho_2^*
\end{equation}
	we 
	have $p_0(R) 
	= 
	p_1(R)$. Moreover, it holds
	\begin{equation}
		p_0(R) = 
		\begin{cases}
			0, &\mathrm{rank}(R) \le 1,\\
			0, &\mathrm{rank}(R) = 2, \Pi_R \neq 
			\proj{\psi} \otimes \1_{\C^2}, \ket{\psi} \in \C^2,\\
			\|\tr_1\left(R^{-1} 
			(\proj{\psi} 
			\otimes \1_{\C^2})\right)\|_\infty^{-1}, &\mathrm{rank}(R) = 2, 
			\Pi_R = 
			\proj{\psi} \otimes \1_{\C^2}, \ket{\psi} \in \C^2,\\
			0, &\mathrm{rank}(R) = 3, \Pi_R = 
			\1_{\C^4} - \proj{\alpha}, \C^4 \ni \ket{\alpha} \mbox{ is 
				entangled},\\
			\|\tr_1\left(R^{-1} 
			(\proj{\psi} 
			\otimes \1_{\C^2})\right)\|_\infty^{-1}, &\mathrm{rank}(R) = 3, 
			\Pi_R = 
			\1_{\C^4} - 
			\proj{\psi^\perp} \otimes 
			\proj{\phi}, \ket{\psi^\perp},\ket{\phi} \in \C^2, \proj{\psi} \in 
			\DD(\C^2).
		\end{cases}
	\end{equation}
\end{non}
\begin{proof}
The proof is based on Lemma~\ref{lem-example-gen}. Let us investigate the value 
of $p_0(R)$. We will consider three cases depending 
on $\mathrm{rank}(R)$. 

In the first case, we assume that $\mathrm{rank}(R) \in \{0,1\}$. Then, for $P$ 
satisfying $\Pi_R (P \otimes X) \Pi_R = P \otimes X$ we have
\begin{equation}
	2 \mathrm{rank}(P) = \mathrm{rank}(P \otimes \1_{\C^2}) = 
	\mathrm{rank}(\Pi_R (P 
	\otimes \1_{\C^2}) \Pi_R) \le \mathrm{rank}(\Pi_R) \le 1.
\end{equation}
Hence, we obtain $\mathrm{rank}(P) \le \frac{1}{2}$ which implies $P = 0$. In 
this case $p_0(R) = 0$.

In the second case, we assume that $\mathrm{rank}(R) = 2$. Using the same 
argumentation for $P$ as in the first case, we get $\mathrm{rank}(P) \le 1$. We 
can write $P = \proj{x}$ for $\ket{x} \in \C^2$.  Note that, if $P \neq 0$, 
then from the equality $\Pi_R \ket{x, y} = \ket{x, y}$ for $\ket{y} \in \C^2$ 
we get $\Pi_R = 
\proj{\psi} 
\otimes \1_{\C^2}$, for $\ket{\psi} = \frac{1}{\|x\|}\ket{x}$. Therefore, if 
for all 
$\proj{\psi} \in 
\DD(\C^2)$ it holds $\Pi_R \neq \proj{\psi} \otimes \1_{\C^2}$, we have $p_0(R) 
= 
0$. 
Otherwise, if $\Pi_R = \proj{\psi_0} \otimes \1_{\C^2}$ for $\proj{\psi_0} \in 
\DD(\C^2)$, we take $P = p \proj{\psi_0}$ for $p \ge 0$. From the assumption 
$p \tr_1\left(R^{-1} (\proj{\psi_0} 
\otimes \1_{\C^2})\right) \le \1_{\C^2}$ we get $p_0(R) = \|\tr_1\left(R^{-1} 
(\proj{\psi_0} 
\otimes \1_{\C^2})\right)\|_\infty^{-1}$.

In the third case, we assume that $\mathrm{rank}(R) = 3$. Again, $P$ can be 
written in the form $P = \proj{x}$ 
for $\ket{x} \in \C^2$. Let $\Pi_R=\1_{\C^4} - \proj{\xi},$ where $\proj{\xi} 
\in 
\DD(\C^4)$. If $P \neq 0$, then 
from the equality $\Pi_R \ket{x, y} = \ket{x, y}$ for $\ket{y} \in \C^2$ we get 
$\braket{\xi}{x,y}=0$, for $\ket{y} \in \C^2$, and hence, $\ket{\xi} 
\propto \ket{x^\perp} \otimes \ket{y}$. Therefore, if $\ket{\xi}$ is entangled, 
we have $p_0(R) = 0$. Otherwise, if $\Pi_R = \1_{\C^4} - \proj{\psi_0^\perp} 
\otimes 
\proj{\phi_0}$ for $\proj{\psi_0^\perp}, \proj{\phi_0} \in 
\DD(\C^2)$, we take $P = p \proj{\psi_0}$ for $p \ge 0$. The assumption 
$p \tr_1\left(R^{-1} (\proj{\psi_0} 
\otimes \1_{\C^2})\right) \le \1_{\C^2}$ implies $p_0(R) = \|\tr_1\left(R^{-1} 
(\proj{\psi_0} 
\otimes \1_{\C^2})\right)\|_\infty^{-1}$.
\end{proof}

\subsection{Proof of 
Proposition~\ref{prop-mixed-encoding}}\label{proof-prop-mixed-encoding}
\begin{non}{Proposition~\ref{prop-mixed-encoding}}
	Let us define an unitary matrix $U \in \UU(\C^4)$ which columns form the 
	magic 
	basis
	\begin{equation}
		U = \frac{1}{\sqrt{2}} \left[\begin{array}{cccc}
			1 & 0 & 0 & i\\
			0 & i & 1 & 0\\
			0 & i & -1& 0\\
			1 & 0 & 0 &-i
		\end{array}\right].
	\end{equation}
	Let us also define a diagonal operator $D(\lambda) \coloneqq 
	\mathrm{diag}^\dagger 
	\left(\lambda\right)$, which is parameterized by 
	a $4-$dimensional real vector $\lambda = (\lambda_1, \lambda_2, \lambda_3, 
	\lambda_4)$, for which it holds $0 < \lambda_i \le 1$. For $R = U 
	D(\lambda) 
	U^\dagger$ and the noise channel $\EE_R$ defined as
	\begin{equation}
		\EE_R(Y) = \proj{0} \otimes \tr_1\left( \sqrt{R} Y \sqrt{R}\right) + 
		\proj{1} \otimes \tr\left([\1_{\C^4} - R] Y\right) \rho_2^*
	\end{equation} we have
	\begin{equation}
		\begin{split}
			p_0(R) &=  \frac{4}{\tr(R^{-1})},\\
			p_1(R) &= \frac{4}{\tr(R^{-1}) + 
			\min\left\{ \left|\frac{1}{\lambda_1} - 
				\frac{1}{\lambda_2} 
				- \frac{1}{\lambda_3} + \frac{1}{\lambda_4}\right|, 
				\left|\left|\frac{1}{\lambda_1}-\frac{1}
				{\lambda_4}\right|
				-\left|\frac{1}{\lambda_2}-
				\frac{1}{\lambda_3}\right|\right| \right\}}.
		\end{split}
	\end{equation}
\end{non}
\begin{proof}
First, we calculate $p_0(R)$. Let $\ket{x} =(x_0, 
x_1)^\top$. Then, we have
\begin{equation}\label{eq-proof-example2}
	(\bra{x} \otimes 
	\1_{\C^2}) R^{-1} (\ket{x} \otimes \1_{\C^2}) =
	\frac{1}{2}\left[\begin{array}{cc}
		\frac{|x_0|^2}{\lambda_1}+ \frac{|x_1|^2}{\lambda_2} + 
		\frac{|x_1|^2}{\lambda_3} + \frac{|x_0|^2}{\lambda_4}& \frac{x_1 
			\bar 
			x_0}{\lambda_1} + \frac{x_0 \bar 
			x_1}{\lambda_2} - \frac{x_0 \bar 
			x_1}{\lambda_3} - \frac{x_1 \bar 
			x_0}{\lambda_4} \\
		\frac{x_0 \bar 
			x_1}{\lambda_1} + \frac{x_1 \bar 
			x_0}{\lambda_2} - \frac{x_1 \bar 
			x_0}{\lambda_3} - \frac{x_0 \bar 
			x_1}{\lambda_4} & \frac{|x_1|^2}{\lambda_1}+ 
		\frac{|x_0|^2}{\lambda_2} + \frac{|x_0|^2}{\lambda_3} + 
		\frac{|x_1|^2}{\lambda_4}
	\end{array}\right].
\end{equation}
We obtain $\tr \left( (\bra{x} \otimes \1_{\C^2}) R^{-1} 
(\ket{x} \otimes \1_{\C^2})\right) 
=\frac{1}{2}\left(\frac{1}{\lambda_1}+\frac{1}{\lambda_2}+\frac{1}
{\lambda_3}+
\frac{1}{\lambda_4}\right) \|x\|_2^2 = \frac12 \tr(R^{-1}) \|x\|_2^2$.
Hence, for any $\rho \in \DD(\C^2)$ we have
$\tr\left(R^{-1} (\rho
\otimes \1_{\C^2})\right)=	\frac12 \tr(R^{-1}) $. Eventually, we obtain the 
following upper bound
\begin{equation}
	\begin{split}
		\|\tr_1\left(R^{-1} 
		(\rho
		\otimes \1_{\C^2})\right)\|_\infty^{-1} \le 2 \left( \tr\left(R^{-1} 
		(\rho
		\otimes \1_{\C^2})\right) \right)^{-1} = 4 
		\left(\tr(R^{-1})\right)^{-1}.
	\end{split}
\end{equation}
That means, $p_0(R) \le 4 \left(\tr(R^{-1})\right)^{-1}.$ To saturate this 
bound, we take the maximally mixed state $\rho = \rho_2^*$ and by using 
Eq.~\eqref{eq-proof-example2} we calculate 
\begin{equation}
	\|\tr_1\left(R^{-1} 
	(\rho_2^*
	\otimes \1_{\C^2})\right)\|_\infty^{-1}= 2 
	\|\tr_1\left(R^{-1}\right)\|_\infty^{-1}=2\left\| \frac12 \tr(R^{-1}) 
	\1_{\C^2}\right\|_\infty^{-1} = 4 \left(\tr(R^{-1})\right)^{-1}.
\end{equation}
Therefore, we showed that $	p_0(R) = 4\left(\tr(R^{-1})\right)^{-1}.$

In the case of $p_1(R)$, to calculate the largest eigenvalue of 
$\tr_1\left(R^{-1} (\proj{x}	\otimes \1_{\C^2})\right)$ we use 
Eq.~\eqref{eq-proof-example2} for $\ket{x} = (|x_0|, |x_1| \alpha)^\top$, such 
that $|x_0|^2+|x_1|^2=1$ and $|\alpha|=1$. One may calculate that the largest 
eigenvalue minimized over $\alpha$ is given by
\begin{equation}
	\begin{split}
		\frac{1}{4}\left(\tr(R^{-1}) + 
		\sqrt{\left(\left(\frac{1}{\lambda_1} + \frac{1}{\lambda_4}\right) - 
		\left( 
			\frac{1}{\lambda_2} 
			+ \frac{1}{\lambda_3}\right) 
	\right)^2(|x_0|^2-|x_1|^2)^2+4\left(\left|\frac{1}{\lambda_1}-\frac{1}
			{\lambda_4}\right|
			-\left|\frac{1}{\lambda_2}-
			\frac{1}{\lambda_3}\right|\right)^2|x_0|^2|x_1|^2}\right).
	\end{split}
\end{equation}
It turns out, there are only two situations when this expression is minimized: 
\begin{itemize}
	\item For $|x_0| = 0$ and $|x_1|=1$ (or equivalently $|x_0|=1$ and 
	$|x_1|=0$), we obtain 
	\begin{equation}
		\begin{split}
			\frac{1}{4}\left(\tr(R^{-1}) + 
			\left|\frac{1}{\lambda_1} - \frac{1}{\lambda_2} 
			- \frac{1}{\lambda_3} + \frac{1}{\lambda_4}\right| \right).
		\end{split}
	\end{equation}
	\item For $|x_0|=|x_1|=\frac{1}{\sqrt{2}}$, we obtain
	\begin{equation}
		\begin{split}
			\frac{1}{4}\left(\tr(R^{-1}) + 
			\left|\left|\frac{1}{\lambda_1}-\frac{1}
			{\lambda_4}\right|
			-\left|\frac{1}{\lambda_2}-
			\frac{1}{\lambda_3}\right|\right| \right).
		\end{split}
	\end{equation}
\end{itemize}

Hence, the optimal value $p_1(R)$ equals
\begin{equation}
	p_1(R) = \frac{4}{\tr(R^{-1}) + \min\left\{\left|\frac{1}{\lambda_1} - 
		\frac{1}{\lambda_2} 
		- \frac{1}{\lambda_3} + \frac{1}{\lambda_4}\right|, 
		\left|\left|\frac{1}{\lambda_1}-\frac{1}
		{\lambda_4}\right|
		-\left|\frac{1}{\lambda_2}-
		\frac{1}{\lambda_3}\right|\right| \right\}}.
\end{equation}
\end{proof}

\subsection{Proof of Proposition~\ref{prop-prop}}\label{proof-prop-prop}
\begin{non}{Proposition~\ref{prop-prop}}
	For any $\XX$, $\YY$ we have the following 
	properties:
	\begin{enumerate}[(A)]
		\item $\xi_1(\XX,\YY) \subset \xi(\XX,\YY),$
		\item If $\dim(\XX) > \dim(\YY)$, then $\xi(\XX,\YY) = \emptyset,$
		\item If $\dim(\XX) \le \dim(\YY)$, then $\xi_1(\XX,\YY) \neq 
		\emptyset,$ 
		\item If $ \dim(\XX) = \dim(\YY)$, then $\xi_1(\XX,\YY) = \xi(\XX,\YY).$
	\end{enumerate}
\end{non} 
\begin{proof}
	$(D)$\\
	Let us take $\EE = \kraus{ (E_i)_i } \in 
	\xi(\XX,\YY)$. From Theorem~\ref{thm-uqec-general} $(D)$ 
	there exist $S_* \in \MM(\XX,\YY)$ and $R_* \in \MM(\YY, 
	\XX) $ such that $R_* E_i S_* \propto \1_\XX$, and there exists $i_0$ for 
	which it holds $R_* E_{i_0} S_* \neq 0$. It implies that $R_*$ and $S_*$ 
	are invertible, so for all $i$ we have $E_i \propto R_*^{-1} S_*^{-1}.$
	Hence, $\mathrm{rank}(J(\EE))=1$, so we can write $\EE(X) = E X 
	E^\dagger$, for $E \in \UU(\XX)$. By taking 
	$\RR = 
	\II_\XX$ and $\SS = \EE^\dagger$ we get $\EE \in 
	\xi_1(\XX,\YY)$.\\
\end{proof}

\subsection{Proof of Theorem~\ref{thm-nowhere}}\label{proof-thm-nowhere}
\begin{non}{Theorem~\ref{thm-nowhere}}
	Let $\XX$ and $\YY$ be Euclidean spaces for which $\dim(\XX) < \dim(\YY)$. 
	Then, 
	the set $\xi_1(\XX,\YY)$ is a nowhere dense subset of $\xi(\XX,\YY)$. 
\end{non}
\begin{proof}
	First, we will prove that $\xi_1(\XX,\YY)$ is a closed set. Define a 
	sequence $(\EE_n)_{n \in \N} \subset 
	\xi_1(\XX,\YY)$ 
	that converges to $\EE = \lim\limits_{n \to \infty} \EE_n \in \CC(\YY).$
	From Proposition~\ref{prop-realization} there exist two sequences 
	$(\SS_n)_{n \in \N} \subset \CC(\XX,\YY)$ and $(\RR_n)_{n \in \N} \subset 
	\CC(\YY,\XX)$ such that $\RR_n \EE_n \SS_n = \II_\XX$ for $n \in \N$. Both 
	sets $\CC(\XX,\YY)$ and 
	$\CC(\YY,\XX)$ are compact, so there exists a subsequence $(n_k)_{k 
		\in 
		\N}$, such that $(\SS_{n_k})_{k \in \N}$, $(\RR_{n_k})_{k \in \N}$ 
	converge to 
	some $\SS \in \CC(\XX,\YY), \RR \in \CC(\YY,\XX)$, respectively. Hence, 
	we obtain $\RR \EE \SS = \lim\limits_{k \to \infty} \RR_{n_k} \EE_{n_k} 
	\SS_{n_k} 
	= \II_\XX$. That ends this part of the proof.
	
	To show that $\xi_1(\XX,\YY)$ is a nowhere dense in $\xi(\XX,\YY)$, it is 
	enough to prove $\mathrm{int}_{\xi(\XX,\YY)}\left(\xi_1(\XX,\YY)\right) = 
	\emptyset$. Therefore, for any $\EE \in \xi_1(\XX,\YY)$ we will construct a 
	sequence of 
	channels 
	$(\EE_n)_{n \in \N} \subset \CC(\YY)$ that converges to $\EE$ and for 
	which $\EE_n \in \xi(\XX,\YY)$, and $\EE_n \not\in \xi_1(\XX,\YY)$, for $n 
	\in \N$.
	
	Fix $\EE \in \xi_1(\XX,\YY)$. From Proposition~\ref{prop-realization} there 
	exist $\SS = \kraus{(S)} \in \CC(\XX,\YY)$ and $\RR \in \CC(\YY,\XX)$ such 
	that 
	$\RR\EE\SS = \II_\XX$. From Theorem~\ref{thm-uqec-general} we have
	\begin{equation}
		\EE \SS = \kraus{(A_i)_i}: \quad A_i \not = 0, A_j^\dagger A_i \propto
		\delta_{ij} \1_\XX.
	\end{equation}
	As $\dim(\XX)< \dim(\YY)$, there exists $\proj{y} \in 
	\DD(\YY)$ such that $\bra{y}A_1=0$. Let us define a sequence of channels 
	$\EE_n \in \CC(\YY)$ given by
	\begin{equation}
		\EE_n(Y) = \frac{n}{n+1}\EE(Y) + \frac{\tr(Y)}{n+1}  \proj{y}.
	\end{equation}
	One can note that $\lim\limits_{n \to \infty} \EE_n = \EE.$ 
	We take $\SS_n = \SS$ and $\RR_n = \kraus{(A_1^\dagger) }$ for $n \in \N$ 
	and obtain
	\begin{equation}
		\RR_n \EE_n \SS_n (X) = \frac{n}{n+1} A_1^\dagger\EE\SS(X)A_1 = 
		\frac{n}{n+1} 
		\|A_1 \|_\infty^4 X.
	\end{equation}
	As $A_1 \ne 0$, it follows that $\EE_n \in \xi(\XX,\YY)$. Now, for each $n 
	\in \N$, let $\widetilde \SS_n \in \CC(\XX, \YY)$ and $\widetilde \RR_n \in 
	s\CC(\YY,\XX)$ be arbitrary operations satisfying $0 \neq \widetilde \RR_n 
	\EE_n  \widetilde \SS_n \propto \II_\XX$.  It holds that $\widetilde \RR_n 
	(\proj{y}) = 0.$ Eventually, for any $\proj{\psi} \in \DD(\XX)$ we have
	\begin{equation}
		\tr\left(\widetilde \RR_n\EE_n \widetilde \SS_n(\proj{\psi})\right) = 
		\frac{n}{n+1}\tr\left(\widetilde\RR_n\EE 
		\widetilde\SS_n(\proj{\psi})\right) \le 
		\frac{n}{n+1}.
	\end{equation}
	Hence, we obtain $\EE_n \not\in \xi_1(\XX,\YY)$. 
\end{proof}

\subsection{Proof of Theorem~\ref{thm-rank}}\label{proof-thm-rank}
\begin{non}{Theorem~\ref{thm-rank}}
	Let $\XX$ and $\YY$ be some Euclidean spaces such that $\dim(\YY) \ge 
	\dim(\XX)$. 
	The following relations hold:
	\begin{equation}
		\begin{array}{llll}
			(A) && \max \left\{ \mathrm{rank}(J(\EE)): \EE \in \xi_1(\XX, \YY) 
			\right\} &= 
			\dim(\YY)^2 
			- \dim(\YY) \dim(\XX) + \floor{\frac{\dim(\YY)}{\dim(\XX)}},\\
			(B) && \max \left\{ \mathrm{rank}(J(\EE)): \EE \in \xi(\XX, \YY) 
			\right\} &= 
			\dim(\YY)^2 
			- \dim(\XX)^2 + 1.
		\end{array}
	\end{equation}
\end{non}
\begin{proof}
	Let us define $d = \dim(\XX)$, $s= \dim(\YY)$ and $k = 
	\floor{\frac{s}{d}}$.\\
	
	$(A)$\\
	Take $\EE = \kraus{(E_i)_{i=1}^r} \in \xi_1(\XX, \YY)$, where $r = 
	\mathrm{rank}(J(\EE))$. From Proposition~\ref{prop-realization} there exist 
	$\SS = \kraus{(S)} \in 
	\CC(\XX,\YY)$ and $\RR 
	\in \CC(\YY,\XX)$ 
	such that $\RR \EE \SS = \II_\XX$. According to 
	Theorem~\ref{thm-uqec-general} it 
	holds
	\begin{equation}
		\kraus{(E_i S)_{i=1}^r} = \kraus{(A_i)_{i=1}^{r'}}: \quad A_i 
		\not = 0, 
		A_j^\dagger A_i 
		\propto \delta_{ij} \1_\XX.
	\end{equation}
	If $r'<r$, then let us define $A_i = 0$ for $i = r'+1,\ldots,r$. There 
	exists the Kraus decomposition $\EE = \kraus{(E_i')_{i=1}^r}$ such 
	that $A_i = E_i' S$ for each $i \le r$. For $A_i \neq 
	0$ images of $A_i$ are orthogonal and $\mathrm{rank}(A_i) = d$. Hence, $r' 
	d \le s$ which is equivalent to $r'	\le k$. For $i > r'$ it holds $ 
	\left(\1_\YY \otimes S^\top \right) \ket{E_i'} = 0$. Note that the Kraus 
	operators $E_i'$ are linearly independent and it holds
	\begin{equation}
		\dim(\mathrm{ker}(\1_\YY \otimes S^\top)) = s^2 - \mathrm{rank}(\1_\YY 
		\otimes S^\top) = s^2 - \mathrm{rank}(\1_\YY) \mathrm{rank}(S) = s^2 
		-sd.
	\end{equation}
	Therefore, we get $r-r' = \dim(\mathrm{span}(E'_i, i > r')) \le 
	\dim(\mathrm{ker}(\1_\YY \otimes S^\top)) = s^2 
	- sd$ and eventually $r \le s^2 - sd +k$. To saturate this bound, let us 
	define $\EE \in 
	\CC(\YY)$ given by
	\begin{equation}
		\EE(Y) = \sum_{i=0}^{k-1} E_i Y 
		E_i^\dagger 
		+ \tr \left( (\1_\YY - \Pi) Y \right) \rho_{\YY}^*,
	\end{equation}
	where
	\begin{equation}
		\begin{split}
			E_i &= 
			\frac{1}{\sqrt{k}}\sum_{j=0}^{d-1}
			\ketbra{j+id}{j} \in \MM(\YY), \quad \quad \mbox{for } 
			i=0,\ldots, 
			k -1,\\
			\Pi &= \sum_{j = 0}^{d-1} \proj{j} \in \PP(\YY).
		\end{split}
	\end{equation}
	Note that $\Pi = \sum_{i=0}^{k-1} E_i^\dagger E_i$ and $(\1_\YY \otimes 
	\Pi) 
	\ket{E_i} = \ket{E_i}$. 
	Therefore, we obtain
	\begin{equation}
		\mathrm{rank}(J(\EE)) = \mathrm{rank}\left(\sum_{i=0}^{k-1} \proj{E_i} 
		+ 
		\rho_{\YY}^* \otimes (\1_\YY - \Pi)\right) = 
		\mathrm{rank}\left(\sum_{i=0}^{k-1} \proj{E_i}\right) +
		\mathrm{rank}\left(\rho_{\YY}^* \otimes (\1_\YY - \Pi)\right) =
		s^2 - 
		sd + k.	
	\end{equation}
	Finally, let us define $\SS = \kraus{(S)} \in \CC(\XX, \YY)$, where $S = 
	\sum_{j=0}^{d-1} 
	\ket{j}_{\YY} \bra{j}_\XX$, and $\RR \in s\CC(\YY,\XX)$ given by $\RR(Y) = 
	k S^\dagger \left(	\sum_{i=0}^{k-1}  
	E_i^\dagger Y 
	E_i \right)S.$ We can observe that $\RR \EE \SS = \II_\XX$, so 
	$\EE \in \xi_1(\XX,\YY)$.\\
	
	$(B)$\\
	Take $\EE = \kraus{(E_i)_{i=1}^r} \in \xi(\XX, \YY)$, where $r = 
	\mathrm{rank}(J(\EE))$. According to 
	Theorem~\ref{thm-uqec-general} $(D)$ there exist $S_* \in \MM(\XX,\YY)$ and 
	$R_* \in \MM(\YY, \XX) $ such that $R_* E_i S_* \propto \1_\XX$, and there 
	exists $i_0$ for which it holds $R_* E_{i_0} S_* \neq 0$. We may assume 
	that $\|R_*\|_\infty\le1$ and $\|S_*\|_\infty \le 1$. Hence, according to 
	Theorem~\ref{thm-uqec-general} $(B)$ we get
	\begin{equation}
		\kraus{ \left( \sqrt{R_*^\dagger R_*} E_i S_* \right)_{i=1}^r} = 
		\kraus{(A_i)_{i = 1}^{r'}}: 
		\quad A_i \not = 0, A_j^\dagger A_i 
		\propto \delta_{ij} \1_\XX.
	\end{equation}
	If $r'<r$, then let us define $A_i = 0$ for $i = r'+1,\ldots,r$. There 
	exists the Kraus decomposition $\EE = \kraus{(E_i')_{i = 1}^r}$ such 
	that $A_i = \sqrt{R_*^\dagger R_*} E_i' S_*$ for each $i \le r$. Let $\Pi$ 
	be the projector on the support of $R_*^\dagger R_*$. Observe that 
	$\mathrm{rank}(\Pi) = d$. Then, for each 
	$i \le r$ we have $\Pi A_i = A_i$ and for $i \le r'$ we have 
	$\mathrm{rank}(A_i) = d$. The relation $A_j^\dagger A_i 
	\propto \delta_{ij} \1_\XX$ implies that there exists exactly one $A_i \neq 
	0$, hence, $r' = 1$. For $i > 1$ we have $ \left(\sqrt{R_*^\dagger R_*} 
	\otimes 
	S_*^\top \right) \ket{E_i'} = 0$. Note that the Kraus operators $E_i'$ are 
	linearly independent and it holds
	\begin{equation}
		\dim\left(\mathrm{ker}\left( \sqrt{R_*^\dagger 
			R_*} \otimes S_*^\top \right)\right) = s^2 - 
		\mathrm{rank}\left( \sqrt{R_*^\dagger 
			R_*} \otimes S_*^\top \right)=s^2 - d^2.
	\end{equation}
	Therefore, we obtain $r-1 = \dim(\mathrm{span}(E'_i, i > 1)) \le 
	\dim\left(\mathrm{ker}\left( \sqrt{R_*^\dagger 
		R_*} \otimes S_*^\top \right)\right) = s^2 - d^2$ and eventually $r 
	\le s^2 - d^2 + 1$. To saturate this bound, we define $\EE \in 
	\CC(\YY)$ 
	given by
	\begin{equation}
		\EE(Y) = \frac{\Pi Y \Pi + \tr \left( \Pi Y 
			\right)(\1_\YY-\Pi)}{s- d +1} 
		+ \tr 
		\left( (\1_\YY - \Pi) Y \right) \rho_{\YY}^*,
	\end{equation}
	where $	\Pi = \sum_{j=0}^{d-1} \proj{j} \in \PP(\YY).$ Note, that 
	\begin{equation}
		\begin{split}
			\mathrm{rank}(J(\EE)) &= \mathrm{rank}\left( \frac{1}{s-d+1} 
			(\proj{\Pi} + (\1_\YY - \Pi) \otimes \Pi) + \rho_{\YY}^* \otimes 
			(\1_\YY - \Pi) \right) \\
			&= \mathrm{rank}(\proj{\Pi}) + \mathrm{rank}((\1_\YY - \Pi) \otimes 
			\Pi) + \mathrm{rank}(\rho_{\YY}^* \otimes 
			(\1_\YY - \Pi) ) = s^2 - d^2 + 1.
		\end{split}
	\end{equation}
	Define $\SS = \kraus{(S)} \in \CC(\XX, \YY)$, where $S = \sum_{j=0}^{d-1} 
	\ket{j}_{\YY} \bra{j}_\XX$	and $\RR \in s\CC(\YY, \XX)$ given by $ \RR(Y) 
	= 
	S^\dagger Y S$. We can observe that $\RR \EE \SS = 
	\frac{\II_\XX}{s-d+1}$, so $\EE \in 
	\xi(\XX,\YY)$.
\end{proof}

\subsection{Proof of Lemma~\ref{lem-rank-bounds}}\label{proof-lem-rank-bounds}
\begin{non}{Lemma~\ref{lem-rank-bounds}}
	Let $\XX$ and $\YY$ be Euclidean spaces such that $\dim(\YY) \ge 
	\dim(\XX)$. Then, there exists a Schur channel $\EE \in \CC(\YY)$ such that 
	$\mathrm{rank}(J(\EE)) = \left\lceil \frac{\dim(\YY)}{\dim(\XX) - 1} 
	\right\rceil$ 
	and $\EE \not\in \xi(\XX, 
	\YY)$. 
	Moreover, there exists a Schur channel $\FF \in \CC(\YY)$ such that 
	$\mathrm{rank}(J(\FF)) = \left\lceil\sqrt{\left\lceil 
	\frac{\dim(\YY)}{\dim(\XX) - 1} 
		\right\rceil}\right\rceil$ and $\FF \not\in \xi_1(\XX, 
	\YY)$. Especially, that implies
	\begin{equation}
		\begin{split}
			r(\XX,\YY) &< \frac{\dim(\YY)}{\dim(\XX) - 1},\\
			r_1(\XX,\YY) &< \sqrt{\frac{\dim(\YY)}{\dim(\XX)-1}}.
		\end{split}
	\end{equation}
\end{non}
\begin{proof}
	Let $d = \dim(\XX)$, $s=\dim(\YY)$ and $s = k (d-1) 
	- p$, where $k = \left\lceil 
		\frac{s}{d - 1} \right\rceil$ and $p \in 
		\{0,\ldots,d-2\}$. First, we will show 
		that $r(\XX,\YY) < k$. Define a Schur channel $\EE = 
		\kraus{(E_i)_{i=0}^{k-1}} \in 
		\CC(\YY)$ given by
\begin{equation}\label{eq-proof-proposition}
\begin{split}
E_i &= \sum_{j = 0}^{d - 2} \proj{j+(d-1)i}, \quad i=0,\ldots,k-2,\\
E_{k-1} &= \sum_{j=0}^{d-2-p} \proj{j+(d-1)(k-1)}.
\end{split}
\end{equation}
Observe that $\mathrm{rank}(J(\EE)) = k$. From Theorem~\ref{thm-uqec-general} 
$(D)$ we know that $\EE \in \xi(\XX, \YY)$ if and only if there exist $S_* \in 
\MM(\XX,\YY)$ and $R_* \in \MM(\YY, 
	\XX) $, such that $	R_* E_i S_* \propto \1_\XX$ for all $i$ 
and there exists $i_0$ for which it holds $R_* E_{i_0} S_* \neq 0$. As 
$\mathrm{rank}(E_i) \le d-1$, if we have $	R_* E_i S_* \propto \1_\XX$, 
then $ R_* E_i S_* = 0$ for all $i$. That implies $\EE \not\in \xi(\XX, 
\YY)$.

Now, let us define $l = \left\lceil\sqrt{k}\right\rceil$. We will prove that 
$r_1(\XX,\YY) < l$. Due to the relation $\mathrm{span}_\C \left(\proj{\psi}: 
\proj{\psi} \in \DD(\C^l) \right) = \MM(\C^l)$, we may define unit vectors 
$\ket{\psi_a}$, for $a=0,\ldots,l^2-1$, such that $\mathrm{span}_\C 
\left(\{\proj{\psi_a} \}_a \right) = \MM(\C^l).$ Let us define $F_i \in 
\MM(\YY)$ for 
$i=0,\ldots,l-1$ given by
\begin{equation}
F_i = \sum_{a=0}^{k-1} \braket{\psi_a}{i} E_a,
\end{equation}
for $E_a$ defined in Eq.~\eqref{eq-proof-proposition}. Observe that $F_i$ are 
linearly independent. We have that
\begin{equation}
\sum_{i=0}^{l-1} F_i^\dagger F_i = \sum_{i=0}^{l-1} \sum_{a,b=0}^{k -1} 
\braket{i}{\psi_b} \braket{\psi_a}{i} E_b^\dagger E_a =  \sum_{i=0}^{l-1} 
\sum_{a=0}^{k -1} 
\braket{i}{\psi_a} \braket{\psi_a}{i}  E_a = \sum_{a=0}^{k -1} E_a = \1_\YY.
\end{equation}
Now, we introduce a Schur channel $\FF = \kraus{(F_i)_{i=0}^{l-1}} 
\in \CC(\YY)$. Assume indirectly that $\FF \in \xi_1(\XX,\YY)$. Then, 
according to Proposition~\ref{prop-realization} and 
Theorem~\ref{thm-uqec-general} there exists $S \in \MM(\XX, \YY)$, which 
satisfies $S^\dagger S = \1_\XX$ and $M \in \MM(\C^l)$, such that $
S^\dagger F_j^\dagger F_i S = M_{ji}\1_\XX.$ Therefore, we get
\begin{equation}
\begin{split}
M \otimes \1_\XX = \sum_{j,i} \ketbra{j}{i} \otimes S^\dagger F_j^\dagger F_i 
S 
= (\1 \otimes S^\dagger)\sum_{j,i} \left( \ketbra{j}{i} \otimes 
\sum_{a=0}^{k-1} 
\braket{j}{\psi_a} \braket{\psi_a}{i}  E_a \right) (\1 \otimes 
S)=\sum_{a=0}^{k-1} \proj{\psi_a} \otimes S^\dagger E_a S.
\end{split}
\end{equation}
For each $a=0,\ldots,k-1$ we can use Gram-Schmidt orthogonalization to define 
$X_a$, such that $\tr(X_a \proj{\psi_a}) \neq 0$ and 
$\tr(X_a \proj{\psi_b}) = 0$ whenever $a \neq b$. Hence, we obtain 
$\tr(X_a M) \1_\XX = \tr(X_a \proj{\psi_a}) S^\dagger E_a S$. As 
$\mathrm{rank}(E_a) \le d-1$ we get $S^\dagger E_a S = 0$ for all $a$. It 
implies that $0 = \sum_{a=0}^{k-1} S^\dagger E_a S = S^\dagger S = \1_\XX,$
which gives the contradiction. That means $\FF \not\in \xi_1(\XX,\YY)$. It is 
enough to observe that $r_1(\XX,\YY) < \mathrm{rank}(J(\FF)) = l$.
\end{proof}

\subsection{Proof of Proposition~\ref{prop-diag}}\label{proof-prop-diag}
\begin{non}{Proposition~\ref{prop-diag}}
	Let $\XX$ and $\YY$ be Euclidean spaces and $\dim(\XX) \le \dim(\YY)$. For 
	any Schur channels $\EE \in \CC(\YY)$, such that $\mathrm{rank}(J(\EE)) < 
	\frac{\dim(\YY)}{\dim(\XX) - 1}$, it 
	holds $\EE \in \xi(\XX, \YY)$.
\end{non}
\begin{proof}
	Let $\Delta \in \CC(\YY)$ be 
	the maximally dephasing channel, 
	that 
	is $\Delta(Y) = \sum_i \proj{i}Y\proj{i}$.
Let us fix $r$ such that $r < \frac{\dim(\YY)}{\dim(\XX) - 1}$. We will show 
that if $\EE 
= \kraus{(E_i)} \in \CC(\YY)$, such that $E_i = \Delta(E_i)$ for 
	each $i$ and $\mathrm{rank}(J(\EE)) \le r $, then $\EE \in \xi(\XX, \YY)$.
 Observe that the thesis is true in two particular situations:
\begin{itemize}
\item For $\dim(\XX) = 1$ and $\dim(\YY) \ge 1$.
\item For $r = 1$ and $\dim(\YY) \ge \dim(\XX)$.
\end{itemize}

Let us take $\EE = \kraus{(E_i)} \in \CC(\YY)$, such that
	$\mathrm{rank}(J(\EE)) \le r$ and $E_i = \Delta(E_i)$ for each $i$.
We may assume that $\mathrm{rank}(J(\EE)) = r$. Therefore, there exists 
a projector $\Pi \in \PP(\YY)$, such that $\mathrm{rank}(\Pi) = r$ and 
$\Delta(\Pi) = \Pi$, and for which the operators $\Pi E_i \Pi$ are linearly 
independent. Let us consider the operation $\FF = 
\kraus{(\Pi^\perp E_i \Pi^\perp)_{i=1}^r}$. Define $\XX' = \C^{\dim(\XX)-1}$. 
By the recurrence and Theorem~\ref{thm-uqec-general} for operation $\FF$ there 
exist $S_*' \in 
\MM\left(\XX', \YY\right)$ and $R_*' \in \MM\left(\YY, \XX'\right)$, such that 
$R_*' \Pi^\perp 
E_i \Pi^\perp S_*' = c_i \1_{\XX'}$ and $c_{i_0} \neq 0 $ for some $i_0$. Let 
$\ket{s} \in \CC(\YY)$ be the flat superposition. As $\Pi E_i \Pi$ are diagonal 
and linearly independent, there exists the vector $\ket{r}$ such that $\bra{r} 
\Pi E_i \Pi \ket{s} = c_i$. We may define an encoding operator $S_*$ by adding 
a column $\Pi \ket{s}$ to the operator $\Pi^\perp 
S_*'$. In the same manner, we may construct $R_*$ by adding a row $\bra{r} \Pi$ 
to the operator 
$R_*' \Pi^\perp$. It is easy to check that $S_*, R_*$ satisfy Theorem 
\ref{thm-uqec-general} $(D)$, so $\EE \in \xi(\XX, \YY)$.
\end{proof}

\subsection{Proof of Proposition~\ref{prop-bi-lin}}\label{proof-prop-bi-lin}
\begin{non}{Proposition~\ref{prop-bi-lin}}
	Let $\XX$ and $\YY$ be some Euclidean spaces and $\dim(\XX) \le \dim(\YY)$.
	\begin{enumerate}[(A)]
		\item If $\EE \in \CC(\YY)$ is a noise channel such that 
		$\mathrm{rank}(\EE(\1_\YY)) 
		= \dim(\XX)$ and $\mathrm{rank}(J(\EE)) < 
		\frac{\dim(\YY)\dim(\XX)}{\dim(\XX)^2 - 1}$, then 
		$\EE \in \xi(\XX, \YY)$.
		\item There exists a noise channel $\EE \in 
		\CC(\YY)$ such that $\mathrm{rank}(\EE(\1_\YY)) = \dim(\XX)$ and 
		$\mathrm{rank}(J(\EE)) \ge 
		\frac{\dim(\YY)\dim(\XX)}{\dim(\XX)^2 - 1}$, for which 
		we have $\EE \not\in \xi(\XX, \YY)$.
	\end{enumerate} 
\end{non}
\begin{proof}
$(A)$\\
Let us take $\EE = \kraus{(E_i)_{i=1}^r} \in 
\CC(\YY)$, where $r = \mathrm{rank}(J(\EE))$. Assume that 
$\mathrm{rank}(\EE(\1_\YY)) = \dim(\XX)$ and 
$r < \frac{\dim(\YY)\dim(\XX)}{\dim(\XX)^2 - 1}$. We can 
consider the equivalent form of the problem by taking the associated channel 
$\FF 
= \kraus{(F_i)_{i=1}^r} \in \CC(\YY, \XX)$. Therefore, $\EE \in \xi(\XX, \YY)$ 
if and only 
if there exists $S \in \MM(\XX, \YY)$ such that $F_i S = c_i \1_\XX$ and 
$c_{i_0} \neq 0$ for some $i_0$. Let $F =  \sum_{i=1}^r \ket{i} \otimes 
F_i \in \MM(\YY, \C^{r} \otimes \XX)$ and $\ket{c} = \sum_{i=1}^r 
c_i\ket{i}$. Hence, $\EE \in \xi(\XX, \YY)$ if and only if it holds 
$FS = \ket{c} \otimes \1_\XX \neq 0,$ which is equivalent to
\begin{equation}
	(F \otimes \1_\XX) \ket{S} = \ket{c} \otimes \ket{\1_\XX} \neq 0.
\end{equation}
As $\mathrm{rank}(F) = \dim(\YY)$, the subspace $\{(F \otimes \1_\XX) \ket{S}: 
\ket{S}\}$ has the dimension $\dim(\YY)\dim(\XX)$. On the other hand, the 
subspace 
$\{\ket{c} \otimes \ket{\1_\XX}: \ket{c}\}$ has the dimension $r$. Therefore, 
as 
long as
\begin{equation}
	\dim(\YY) \dim(\XX) + r > r	\dim(\XX)^2
\end{equation}
there exists non-zero solution $S \in \MM(\XX,\YY)$ and $\ket{c} \in \C^r$, 
such that $(F \otimes \1_\XX) \ket{S} = \ket{c} \otimes \ket{\1_\XX}$. From the 
inequality $r < \frac{\dim(\YY)\dim(\XX)}{\dim(\XX)^2 - 1}$ we obtain $\EE \in 
\xi(\XX, \YY)$.\\

$(B)$\\
In the part $(A)$ of the proof we showed that
\begin{equation}
\EE \not\in \xi(\XX, \YY) \iff \left((F \otimes \1_\XX) \ket{S} = \ket{c} 
\otimes \ket{\1_\XX} \implies \ket{S} = 0\right).
\end{equation}
Therefore, in this proof, we will construct appropriate operator $F$, such that 
the latter condition holds. It would imply that the associated channel $\EE$ is 
not 
probabilistically correctable. Formally, the operator $F$ should be an isometry 
operator, but by Lemma~\ref{lem-weird-channels}, it is enough 
to define $F$ such that $\mathrm{rank}(F) = \dim(\YY)$.

Let $d = \dim(\XX), s = \dim(\YY)$ and fix $r \in \N$, such that $r \ge 
\frac{sd}{d^2 - 1}$. We start with the case $s = kd$ 
for $k \in \N$. Consider the decomposition $F = 
\sum_{i=0}^{r-1} \ket{i} \otimes F_i$, where $F_i \in \MM(\YY, \XX)$. For $i = 
0,\ldots,k-1$ we define
\begin{equation}
	F_i = \bra{i} \otimes \1_\XX.
\end{equation}
Let $\left\{\1_\XX, (M_j)_{j=0}^{d^2-2}\right\} \subset \MM(\XX)$ be a basis 
of  $\MM(\XX)$. For each $i = k, \ldots, r-1$ we define
\begin{equation}\label{eq-proof-bi-lin}
	F_i = \sum_{j=0}^{d^2-2} \delta(j + (i-k)(d^2-1) < k)\bra{j + (i-k)(d^2-1) 
	} 	\otimes M_j.
\end{equation}
Observe, that $\mathrm{rank}(F) = s$. Let us take $S$ which satisfies $F_iS 
\propto \1_\XX$ for each $i$. Basing on the equations with indices $i = 0, 
\ldots, k-1$ we get $S = \ket{c} \otimes \1_\XX$ for some $\ket{c} = 
\sum_{j=0}^{k-1} c_j \ket{j}$. Note, that if for any $i 
= k, \ldots, r-1$ it holds $F_i S \propto \1_\XX$, then 
$c_j = 0$ for each $j = (i-k)(d^2-1),\ldots, d^2 - 2 + (i-k)(d^2-1)$.
From the assumption $r \ge \frac{sd}{d^2 - 1}$ we have $(r-k)(d^2-1) \ge k$, 
hence, all entries $c_j$ are zeroed. It implies $S=0$.

The case $s = kd + l$ for $ l =1,\ldots, d-1$ is more technically engaging than 
the previous case but it is based on the same idea. It will be only 
briefly discussed. For $i=0,\ldots,k-1$ we can define $F_i$ similarly as in the 
previous case, that is $F_i \sim \bra{i} \otimes \1_\XX$. The operator $F_k$ 
has 
a special form, $F_k \sim (\bra{k} \otimes \sum_{j=0}^{l-1} \proj{j}) + N$, 
where the image of $N$ is contained in $\mathrm{span}(\ket{j}: j \ge l).$ Here, 
the operator $S$ which satisfy $F_iS\propto \1_\XX$ has the form $S \sim  
\ket{c} 
\otimes 
\1_\XX$ for some $\ket{c} = \sum_{j=0}^{k} c_j \ket{j}$. We can choose $N$ such 
that $d(d-l)$ entries $c_j$ will be zeroed if $F_k S \propto \1_\XX$. Finally, 
operators $F_i$ for $i=k+1,\ldots,r-1$ has the analogous form as 
Eq.~\eqref{eq-proof-bi-lin} -- each nullify $(d^2-1)$ entries. In total, the 
number of entries $c_j$ which can be zeroed is not less than $k+1$. Indeed, it 
holds
\begin{equation}
d(d-l) + (r-k-1)(d^2-1) \ge k+1.
\end{equation}
Therefore, $S = 0$, which ends the proof.
\end{proof}

\subsection{Proof of Theorem~\ref{thm-best-bounds}}\label{proof-thm-best-bounds}
\begin{non}{Theorem~\ref{thm-best-bounds}}
	Let $\XX$ and $\YY$ be some Euclidean spaces such that $\dim(\YY) \ge 
	\dim(\XX)$. Then, we have
	\begin{equation}
		\begin{split}
			\left\lfloor \sqrt[4]{\frac{\dim(\YY)}{\dim(\XX)}} 
			\right\rfloor \le r_1(\XX,\YY) \le 		
			\left\lceil\sqrt{\frac{\dim(\YY)}{\dim(\XX)-1}}\right\rceil-1\le 
			r(\XX,\YY)
			< 
			\frac{\dim(\YY)\dim(\XX)}{\dim(\XX)^2 - 1}.
		\end{split}
	\end{equation}
\end{non}
\begin{proof}
	The inequality $ \left\lfloor \sqrt[4]{\frac{\dim(\YY)}{\dim(\XX)}} 
	\right\rfloor \le r_1(\XX,\YY)  $ follows directly from 
	\cite{knill2000theory}. The 
	inequalities $r_1(\XX,\YY) \le 		
	\left\lceil\sqrt{\frac{\dim(\YY)}{\dim(\XX)-1}}\right\rceil-1$ and 
	$r(\XX,\YY) < \frac{\dim(\YY)\dim(\XX)}{\dim(\XX)^2 - 1}$ follow from 
	Lemma~\ref{lem-rank-bounds} and 
	Proposition~\ref{prop-bi-lin}, respectively. 

Now, we will show that 
$\left\lceil\sqrt{\frac{\dim(\YY)}{\dim(\XX)-1}}\right\rceil-1\le r(\XX,\YY)$. 
Take arbitrary $\EE \in \CC(\YY)$ such that 
$\mathrm{rank}(J(\EE))^2(\dim(\XX)-1) < \dim(\YY)$. We will show $\EE \in 
\xi(\XX, \YY)$. Let us denote $r = \mathrm{rank}(J(\EE))$. Consider a Kraus 
representation $\EE = 
\kraus{(E_j)_{j=1}^r}$ and define the following set
\begin{equation}
A = \left\{ s \in \N: \,\, \exists_{\Pi_s \in 
\PP(\YY)} \,\, \Pi_s = \Pi_s^2, \mathrm{rank}(\Pi_s)=s, 
\mathrm{rank}(\EE^\dagger(\Pi_s))=\dim(\YY)  \right\}.
\end{equation}
Observe that $\dim(\YY) \in A$ and if some $s \in A$, then $sr \ge \dim(\YY)$. 
Define $s_0 = \min(A)$ and consider a corresponding projector $\Pi_{s_0} \in 
\PP(\YY)$, such that $\mathrm{rank}(\Pi_{s_0})=s_0$ and 
$\mathrm{rank}(\EE^\dagger(\Pi_{s_0}))=\dim(\YY)$. Let us take a orthonormal 
collection of vectors $\ket{v_i}$, where $i=1,\ldots,s_0$ for which we have 
$\Pi_{s_0} = \sum_{i=1}^{s_0} \proj{v_i}$. From the assumption $s_0 = \min(A)$, 
for any $i$ we get $\mathrm{rank}(\EE^\dagger(\Pi_{s_0}-\proj{v_i})) < 
\dim(\YY)$. Therefore, we may define vectors $0 \neq \ket{w_i} \in \YY$ such 
that 
$\EE^\dagger(\Pi_{s_0}-\proj{v_i})\ket{w_i} = 0$. Observe that for each $i$, 
there 
exists $E_j$ for which $\bra{v_i}E_j\ket{w_i} \neq 0$. Let us define $F_j 
= [\bra{v_a}E_j\ket{w_b}]_{a,b=1,\ldots,s_0}$ for 
$j=1,\ldots,r$. Note, that $F_j$ are diagonal operators and it holds $\sum_j 
F_j^\dagger F_j > 0$. From $r^2(\dim(\XX)-1) < \dim(\YY)$ and 
$s_0r \ge \dim(\YY)$ we have
\begin{equation}
r(\dim(\XX)-1) < \frac{\dim(\YY)}{r}\le s_0.
\end{equation}
Utilizing Proposition~\ref{prop-diag}, Lemma~\ref{lem-weird-channels} and 
Theorem~\ref{thm-uqec-general} there exist $S_* \in \MM(\XX,\C^{s_0})$ and $R_* 
\in \MM(\C^{s_0}, \XX) $, such that $R_* F_j S_* \propto \1_\XX$ and there 
exists $j_0$, for which it holds $R_* F_{j_0} S_* \neq 0$. That implies $\EE 
\in \xi(\XX, \YY)$.
\end{proof}

\subsection{Proof of 
Proposition~\ref{prop-qubit-rank}}\label{proof-prop-qubit-rank}
\begin{non}{Proposition~\ref{prop-qubit-rank}}
	For all $\EE \in \CC(\C^4)$ satisfying $\mathrm{rank}(J(\EE)) \le 2$ we 
	have $\EE \in \xi(\C^2, \C^4)$.
\end{non}
\begin{proof}
Let us fix $\EE = \kraus{(E_0, E_1)} \in \CC(\C^4)$. From the equality 
$E_0^\dagger E_0 + E_1^\dagger E_1 = \1_{\C^4}$ we may write the singular 
decomposition of $E_0, E_1$ in the form: $E_0 = U_0 D_0 V$ and $E_1 = U_1 D_1 
V$, where $U_0, U_1, V \in \UU(\C^4)$ and $D_0, D_1 \in 
\PP(\C^4)$ are 
diagonal operators satisfying $D_0^2+D_1^2=\1_{\C^4}$. In order to show that 
$\EE 
\in \xi(\C^2, \C^4)$ we will use Theorem~\ref{thm-uqec-general} $(D)$. We will 
prove that there exist $S_* \in \MM(\C^2, \C^4)$ and $R_* \in \MM(\C^4, \C^2)$, 
such that $R_* E_0 S_* = c_0 \1_{\C^2}$, $R_* E_1 S_* = c_1 \1_{\C^2}$ for some 
$c_0, c_1 
\in \C$ satisfying $(c_0, c_1) \neq (0, 0)$. Let us introduce the 
following notation
\begin{equation}
	\begin{split}
		\ket{x_i} &= (D_0)_{ii} U_0\ket{i}, \quad i=0,\ldots,3,\\
		\ket{y_i} &= (D_1)_{ii} U_1\ket{i}, \quad i=0,\ldots,3.\\
	\end{split}
\end{equation}
Note 
that vectors $\ket{x_i}$ are orthogonal (the same holds 
for $\ket{y_i}$) and for each $i=0,\ldots,3$ we have $\ket{x_i} \neq 0$ or 
$\ket{y_i} \neq 0$. We may write $S_*$ and $R_*$ in the following form
\begin{equation}
\begin{split}
S_* &= V^\dagger (\ketbra{S_0}{0} + \ketbra{S_1}{1}),\\
R_* &= 	\ketbra{0}{R_0} + \ketbra{1}{R_1},
\end{split}	
\end{equation}
for some vectors $\ket{S_0}, \ket{S_1}, \ket{R_0}, \ket{R_1} \in \C^4$. The 
rest of the prove will be divided into three cases.  

In the first case, we assume there exists $i_3 \in \{0,\ldots,3\}$ such 
that vectors $\ket{x_{i_3}}, \ket{y_{i_3}}$ are linearly independent.
Define indices $i_0, i_1, i_2 \in \{0,\ldots,3\}$ as the remaining labels, 
such that $\{i_0,\ldots,i_3\}$ covers the whole set $\{0,\ldots,3\}$. Let 
$(a_0, a_1, a_2)^\top \in \C^3$ be a normalized vector orthogonal to vectors 
$(\braket{y_{i_3}}{x_{i_0}}, \braket{y_{i_3}}{x_{i_1}}, 
\braket{y_{i_3}}{x_{i_2}})^\dagger$ and $(\braket{x_{i_3}}{y_{i_0}}, 
\braket{x_{i_3}}{y_{i_1}}, \braket{x_{i_3}}{y_{i_2}})^\dagger$. Take 
$\ket{S_1} = 
\ket{i_3}$ and $\ket{S_0} = a_0\ket{i_0} + a_1 \ket{i_1} + a_2 \ket{i_2}$.
Define $\ket{ x} = a_0 \ket{x_{i_0}} + a_1 
\ket{x_{i_1}} + a_2 \ket{x_{i_2}}$ and $\ket{ y} = a_0 
\ket{y_{i_0}} + a_1 
\ket{y_{i_1}} + a_2 
\ket{y_{i_2}}$. We obtain
\begin{equation}
	\begin{split}
		E_0 S_* &= \ketbra{x}{0} + 
		\ketbra{x_{i_3}}{1}, \\
		E_1 S_* &= \ketbra{y}{0} + 
		\ketbra{y_{i_3}}{1}.
	\end{split} 
\end{equation} 
It is not hard to observe that $\ket{ x} \neq 0$ or $\ket{y} \neq 0$. If $\ket{ 
x} \neq 0$, take $\ket{R_0} = 
\ket{ x}$, else take $\ket{R_0} = \ket{y} $.  As the vectors $\ket{x_{i_3}}, 
\ket{y_{i_3}}$ are linearly independent we may define
\begin{equation}
	(b_0, b_1)^\top \coloneqq \left[\begin{array}{cc}
		\braket{x_{i_3}}{x_{i_3}} & \braket{y_{i_3}}{x_{i_3}} \\
		\braket{x_{i_3}}{y_{i_3}} & \braket{y_{i_3}}{y_{i_3}}
	\end{array}\right]^{-1} (\braket{R_0}{ x}, \braket{R_0}{ 
		y})^\top.
\end{equation}
Take $\ket{R_1} = \bar b_0 \ket{x_{i_3}} + \bar b_1 \ket{y_{i_3}}$.
Eventually, we may check that it holds
\begin{equation}
	\begin{split}
		R_* E_0 S_* &= (\ketbra{0}{R_0} + \ketbra{1}{R_1})(\ketbra{ 
			x}{0} + 
		\ketbra{x_{i_3}}{1}) =  \braket{R_0}{ x} \1_{\C^2}, \\
		R_* E_1 S_* &= (\ketbra{0}{R_0} + \ketbra{1}{R_1})(\ketbra{ 
			y}{0} + 
		\ketbra{y_{i_3}}{1}) = \braket{R_0}{ 
			y} \1_{\C^2}.
	\end{split}
\end{equation}

In the second case, we assume that there exists a pair of vectors 
$\ket{y_{i_0}}, \ket{y_{i_1}}$ for $i_0 \neq i_1$, such that
$\ket{y_{i_0}} = \ket{y_{i_1}} = 
0$. Then, the vectors $\ket{x_{i_0}}, \ket{x_{i_1}}$ are orthonormal. We 
simply define $\ket{S_0} = \ket{i_0}, \ket{S_1} = \ket{i_1}$, $\ket{R_0} = 
\ket{x_{i_0}}$ and $\ket{R_1} = \ket{x_{i_1}}$. 
One can calculate that $R_* E_0  S_* = \1_{\C^2}$ and $R_* E_1 S_* = 0$.

In the third case, for all $i \in \{0,\ldots,3\}$ vectors 
$\ket{x_i}, 
\ket{y_i}$ are not linearly independent and there is at most one zero 
vector $\ket{y_{i_3}}$ for some $i_3 \in \{0,\ldots,3\}$.
Define indices 
$i_0, i_1, i_2 \in \{0,\ldots,3\}$ as the remaining labels, such that 
$\{i_0,\ldots,i_3\}$ covers the whole set $\{0,\ldots,3\}$. Define the matrix
\begin{equation}
	M = \left[\begin{array}{ccc}
		\braket{y_{i_0}}{x_{i_0}} & \braket{y_{i_1}}{x_{i_1}} & 
		\braket{y_{i_2}}{x_{i_2}}\\
		\braket{y_{i_0}}{y_{i_0}} & \braket{y_{i_1}}{y_{i_1}} & 
		\braket{y_{i_2}}{y_{i_2}}
	\end{array}\right].
\end{equation}

In the first sub-case we assume that $\mathrm{rank}(M) = 1$. Define $b = 
\frac{\braket{y_{i_1}}{y_{i_1}}}{\braket{y_{i_0}}{y_{i_0}}}$. We can take
$\ket{S_0}=\ket{i_0}$, $\ket{S_1} = \ket{i_1}$, $\ket{R_0} = 
\ket{y_{i_0}}$ and $\ket{R_1} = \frac{1}{b}\ket{y_{i_1}}$. 
One can calculate that $R_* E_0 S_* = 
\braket{y_{i_0}}{x_{i_0}} \1_{\C^2}$ and $R_* E_1 S_* = 
\braket{y_{i_0}}{y_{i_0}} 
\1_{\C^2}$. 

In the second sub-case we assume that $\mathrm{rank}(M) = 2$. Define indices 
$j_1, j_2 \in \{0,1,2\}$, such that 
\begin{equation}
	\mathrm{rank}\left(\left[\begin{array}{cc}
		M_{0,j_1} & M_{0,j_2}\\
		M_{1,j_1} & M_{1,j_2}
	\end{array}\right]\right)=2.
\end{equation}
Define $j_0 \in \{0,1,2\}$ as the remaining label, such that $\{j_0, 
j_1, j_2\}$ covers the whole set $\{0,1,2\}$. Take $\ket{S_0} = \ket{i_{j_0}}$, 
$\ket{R_0} = \ket{y_{i_{j_0}}}$ and define
\begin{equation}
	(b_1, b_2)^\top \coloneqq \left[\begin{array}{cc}
		\braket{y_{i_{j_1}}}{x_{i_{j_1}}} & 
		\braket{y_{i_{j_2}}}{x_{i_{j_2}}}\\
		\braket{y_{i_{j_1}}}{y_{i_{j_1}}} & 
		\braket{y_{i_{j_2}}}{y_{i_{j_2}}}
	\end{array}\right]^{-1} (\braket{y_{i_{j_0}}}{x_{i_{j_0}}}, 
	\braket{y_{i_{j_0}}}{y_{i_{j_0}}})^\top. 
\end{equation}
We may take $\ket{S_1} = \ket{i_{j_1}} + \ket{i_{j_2}}$ and 
$\ket{R_1} = \bar b_1 \ket{y_{i_{j_1}}} + \bar b_2 \ket{y_{i_{j_2}}}$. 
Direct calculations reveal that $R_* E_0 S_* = 
\braket{y_{i_{j_0}}}{x_{i_{j_0}}} \1_{\C^2}$ and $R_* E_1 S_* = 
\braket{y_{i_{j_0}}}{y_{i_{j_0}}} \1_{\C^2}$.
\end{proof}

\subsection{Proof of Theorem~\ref{thm-random}}\label{proof-thm-random}
\begin{non}{Theorem~\ref{thm-random}}
	Let $\EE_r \in \CC(\YY)$ be a random quantum channel 
	defined according to Eq.~\eqref{eq-random}. Then, the following two 
	implications hold
	\begin{equation}
		\begin{split}
			r < \frac{\dim(\XX) \dim(\YY)}{\dim(\XX)^2 - 1} &\implies \PP\left( 
			\EE_{r} \in 
			\xi(\XX, \YY)\right) = 1,\\
			\PP\left( \EE_{r} \in \xi_1(\XX, \YY)\right) = 1 &\implies r < 
			\sqrt{\frac{\dim(\YY)}{\dim(\XX)-1}}.
		\end{split}
	\end{equation}
\end{non}
\begin{proof}
	For $r \in \N$ satisfying $r < \frac{\dim(\XX) \dim(\YY)}{\dim(\XX)^2 - 
	1}$, let 
	$(G_i)_{i=1}^r \subset \MM(\YY)$ be a tuple of random and independent 
	Ginibre 
	matrices and $Q 
	= \sum_{i=1}^r G_i^\dagger G_i$. Define the projector $\Pi  = 
	\sum_{i=0}^{\dim(\XX) - 1} \proj{i}$ and consider the set
	\begin{equation}
		A = \left\{ (G_i)_{i=1}^r: \quad \mathrm{rank}(Q) = \dim(\YY), 
		\mathrm{rank}\left(\sum_{i=1}^r 
		G_i^\dagger \Pi G_i\right) = \min\{r\dim(\XX), \dim(\YY)\}  \right\}.
	\end{equation}
	One can observe that $\PP((G_i)_{i=1}^r \in A) = 1$. Let $\EE_r \in 
	\CC(\YY)$ be a random channel defined according to Eq.~\eqref{eq-random} 
	for $(G_i)_{i=1}^r \in A$, that is $\EE_r(Y) = \sum_{i=1}^r 
	\left(G_iQ^{-1/2}\right) Y 	\left(G_iQ^{-1/2}\right)^\dagger.$
	Define $S = Q^{1/2} \tilde{S}$ for $\tilde{S} \in \MM(\XX, \YY)$ and $R = 
	\tilde{R} \Pi $ for $\tilde{R} \in \MM(\YY, \XX)$. We obtain
	$R G_iQ^{-1/2} S = \tilde{R} \Pi G_i \tilde{S}.$ Utilizing 
	Lemma~\ref{lem-weird-channels}, Proposition \ref{prop-bi-lin} and 
	Theorem~\ref{thm-uqec-general} $(D)$ for 
	$\tilde{\EE} = \kraus{(\Pi G_i)_{i=1}^r} \in s\CC(\YY)$, there exist 
	$\tilde S, 
	\tilde R$, such that $\tilde{R} \Pi G_i \tilde{S} 
	\propto \1_\XX$ and $\tilde{R} \Pi G_{i_0} \tilde{S} \neq 0$ for some 
	$i_0$. Eventually, $\EE_r \in \xi(\XX, 
	\YY)$.
	
	Now, for a given $r \in \N$ let us define $B = \{\EE_{r}: \,\, \EE_{r} \in 
	\xi_1(\XX, \YY)\}$. From the assumption $\PP(B) = 1$, we obtain that $B$ is 
	a dense subset of $\{\EE \in 
	\CC(\YY): \,\, \mathrm{rank}(J(\EE)) \le r\}$. Imitating the proof of 
	Theorem~\ref{thm-nowhere}, we get that if $\EE \in \CC(\YY)$ and 
	$\mathrm{rank}(J(\EE)) \le r$, then $\EE \in \xi_1(\XX, \YY)$. That 
	implies $r \le r_1(\XX, \YY)$. By using 
	Lemma~\ref{lem-rank-bounds} we 
	obtain the desired inequality.
\end{proof}

\subsection{Proof of 
Proposition~\ref{prop-noise-model}}\label{proof-prop-noise-model}
\begin{non}{Proposition~\ref{prop-noise-model}}
	Let $\Upsilon \subset \CC(\YY)$ be a nonempty and convex family of noise 
	channels. Define $\mu$ to be a probability measure defined on $\Upsilon$ 
	and 
	assume that the support of $\mu$ is equal to $\Upsilon$. Let $\bar \EE = 
	\int_{\Upsilon} \EE \mu(d\EE) \in \CC(\YY)$ and fix 
	$(\SS, \RR) \in s\CC(\XX,\YY) \times s\CC(\YY,\XX)$. The following 
	conditions 
	are equivalent:
	\begin{enumerate}[(A)]
		\item For each $\EE \in \Upsilon$ there exists $p_\EE \ge 0$ such that 
		$\RR 
		\EE \SS = p_\EE \II_\XX$ and $	\int_{\Upsilon} p_\EE \mu(d\EE) > 0.$
		\item It holds that $0 \neq \RR \bar \EE \SS \propto \II_\XX$.
	\end{enumerate}
\end{non}
\begin{proof}
	$(B) \implies (A)$\\
	Let us assume that $\RR \bar \EE \SS = p \II_\XX$ for $p > 0$. There exists 
	a $k$ dimensional affine subspace $\LL$ such that $\Upsilon \subset 
	\LL$ and 
	$\mathrm{int}_\LL(\Upsilon) \neq \emptyset$. Take 
	arbitrary $\EE_0 \in \Upsilon$. There exist $\EE_1,\ldots,\EE_k \in 
	\Upsilon$ 
	such 
	that convex hull of points $\EE_0,\ldots,\EE_k$ is a $k$-dimensional 
	simplex 
	$\Delta_k$. For any state $\proj{\psi} \in \DD(\XX)$ it holds
	\begin{equation}
		p\proj{\psi} = \RR \bar \EE \SS (\proj{\psi}) = \int_{\Upsilon} \RR 
		\EE \SS 
		(\proj{\psi}) \mu(d\EE) \ge \int_{\Delta_k} \RR \EE \SS 
		(\proj{\psi}) 
		\mu(d\EE) .
	\end{equation}
	Inside $\Delta_k$ each $\EE$ can be uniquely represented as 
	$\sum_{i=0}^k 
	q_i(\EE) \EE_i$, where $(q_i(\EE))_{i=0}^k$ is a probability vector which 
	depends on 
	$\EE$. Hence,
	\begin{equation}
		p\proj{\psi} \ge \sum_{i=0}^k \int_{\Delta_k} q_i(\EE) \RR \EE_i 
		\SS 
		(\proj{\psi}) \mu(d\EE) \ge \left(\int_{\Delta_k} 
		q_0(\EE)\mu(d\EE)\right) 
		\RR \EE_0 
		\SS (\proj{\psi}).
	\end{equation}
	There exists $\epsilon$ small ball $B_\epsilon$ around $\EE_0$, such that 
	for each channel $\EE \in B_\epsilon \cap \Delta_k$ it holds $q_0(\EE) 
	\ge \frac{1}{2}$. Hence, $	\int_{\Delta_k} q_0(\EE)\mu(d\EE) \ge \frac12 
	\mu\left(B_\epsilon \cap \Delta_k\right) > 0,$ 	where in the last 
	inequality we used the fact that the support of $\mu$ is 
	equal to $\Upsilon$. Therefore, it holds that for any $\proj{\psi} \in 
	\DD(\XX)$ we have $\RR \EE_0 
	\SS (\proj{\psi}) \propto \proj{\psi}$ and from 
	Lemma~\ref{lem-constant-p} there exists $p_{\EE_0} \ge 0$ such that 
	$\RR \EE_0 \SS = p_{\EE_0} \II_\XX$. The instant relation $	\int_{\Upsilon} 
	p_\EE \mu(d\EE) = p > 0$ ends the proof.  
\end{proof}
\end{document}